\pdfoutput=1
\documentclass[a4paper,11pt]{article}
\usepackage[utf8]{inputenc}
\usepackage[T1]{fontenc}
\usepackage{graphicx}

\usepackage{xspace}
\usepackage{mathtools}
\usepackage{amsthm}
\usepackage{amssymb}
\usepackage{amsmath} 
\usepackage{physics} 
\usepackage{tensor}
\usepackage{slashed}

\usepackage{xcolor}
\newcommand{\hairsp}{\hspace{1pt}}
\newcommand{\ie}{\mbox{\emph{i.\hairsp{}e.}}\xspace}
\newcommand{\eg}{\mbox{\emph{e.\hairsp{}g.}}\xspace}
\newcommand{\cf}{\mbox{\emph{cf.}}\xspace}
\newcommand{\etc}{\mbox{\emph{etc.}}\xspace}
\newcommand{\eqn}[1]{(\ref{eq:#1})}
\newcommand{\mip}[2]{\langle#1,#2\rangle}
\newcommand{\hilb}{\mathcal{H}}
\newcommand{\alg}[1]{\mathcal{#1}}
\newcommand{\lmod}[2]{\tensor[_{#2}]{#1}{}}
\newcommand{\rmod}[2]{\tensor{#1}{_{#2}}}
\newcommand{\bimod}[3]{\tensor[_{#2}]{#1}{_{#3}}}
\newcommand{\lmip}[3]{\tensor[_{#3}]{\langle#1,#2\rangle}{}}
\NewDocumentCommand{\bdd}{O{\hilb}}{\mathcal{B}(#1)} 
\DeclareMathOperator{\Ad}{Ad}
\DeclareMathOperator{\id}{id}
\DeclareMathOperator{\cc}{c.\!c.}
\DeclareMathOperator{\diag}{diag}
\DeclareMathOperator{\spaceofendomorphisms}{End}
\newcommand{\End}[2]{\spaceofendomorphisms_{#1}(#2)}
\DeclareMathOperator{\Aut}{Aut}

\DeclareMathOperator{\linspan}{span}

\DeclareMathOperator{\Dom}{Dom}
\newcommand{\place}{\mathord{\color{black!33}\bullet}}
\newcommand{\ZZ}{\mathbb{Z}}

\newcommand{\CC}{\mathbb{C}}
\newcommand{\HH}{\mathbb{H}}
\DeclareMathOperator{\U}{U}
\DeclareMathOperator{\SU}{SU}

\newcommand{\ASM}{\mathcal{A}_{\mathrm{SM}}}
\newcommand{\ALR}{\mathcal{A}_{\mathrm{LR}}}
\newcommand{\Atoy}{\mathcal{A}_{\mathrm{toy}}}
\newcommand{\DSM}{D_{\mathrm{SM}}}
\newcommand{\EBA}{\bimod{\alg{E}}{\alg{B}}{\alg{A}}}
\newcommand{\conjEBA}{\bimod{\bar{\alg{E}}}{\alg{A}}{\alg{B}}}

\usepackage{scalerel}
\usepackage{bm}

\theoremstyle{plain}
  \newtheorem{thm}{Theorem}[section]
  \newtheorem{lem}[thm]{Lemma}
  
  \newtheorem{prop}[thm]{Proposition}
  
\theoremstyle{definition}
  \newtheorem{defn/}[thm]{Definition}
  \newtheorem*{notn}{Notation}
  \newtheorem{exm/}[thm]{Example}
\theoremstyle{remark}
  \newtheorem*{rmk/}{Remark}
  
  \newtheorem*{note}{Note}
\newenvironment{rmk}
  {%
   \pushQED{\qed}\begin{rmk/}}
  {\popQED\end{rmk/}}
\newenvironment{defn}
  {%
   \pushQED{\qed}\begin{defn/}}
  {\popQED\end{defn/}}

\makeatletter
\let\save@mathaccent\mathaccent
\newcommand*\if@single[3]{%
  \setbox0\hbox{${\mathaccent"0362{#1}}^H$}%
  \setbox2\hbox{${\mathaccent"0362{\kern0pt#1}}^H$}%
  \ifdim\ht0=\ht2 #3\else #2\fi
  }
\newcommand*\rel@kern[1]{\kern#1\dimexpr\macc@kerna}
\newcommand*\widebar[1]{\@ifnextchar^{{\wide@bar{#1}{0}}}{\wide@bar{#1}{1}}}
\newcommand*\wide@bar[2]{\if@single{#1}{\wide@bar@{#1}{#2}{1}}{\wide@bar@{#1}{#2}{2}}}
\newcommand*\wide@bar@[3]{%
  \begingroup
  \def\mathaccent##1##2{%
    \let\mathaccent\save@mathaccent
    \if#32 \let\macc@nucleus\first@char \fi
    \setbox\z@\hbox{$\macc@style{\macc@nucleus}_{}$}%
    \setbox\tw@\hbox{$\macc@style{\macc@nucleus}{}_{}$}%
    \dimen@\wd\tw@
    \advance\dimen@-\wd\z@
    \divide\dimen@ 3
    \@tempdima\wd\tw@
    \advance\@tempdima-\scriptspace
    \divide\@tempdima 10
    \advance\dimen@-\@tempdima
    \ifdim\dimen@>\z@ \dimen@0pt\fi
    \rel@kern{0.6}\kern-\dimen@
    \if#31
      \overline{\rel@kern{-0.6}\kern\dimen@\macc@nucleus\rel@kern{0.4}\kern\dimen@}%
      \advance\dimen@0.4\dimexpr\macc@kerna
      \let\final@kern#2%
      \ifdim\dimen@<\z@ \let\final@kern1\fi
      \if\final@kern1 \kern-\dimen@\fi
    \else
      \overline{\rel@kern{-0.6}\kern\dimen@#1}%
    \fi
  }%
  \macc@depth\@ne
  \let\math@bgroup\@empty \let\math@egroup\macc@set@skewchar
  \mathsurround\z@ \frozen@everymath{\mathgroup\macc@group\relax}%
  \macc@set@skewchar\relax
  \let\mathaccentV\macc@nested@a
  \if#31
    \macc@nested@a\relax111{#1}%
  \else
    \def\gobble@till@marker##1\endmarker{}%
    \futurelet\first@char\gobble@till@marker#1\endmarker
    \ifcat\noexpand\first@char A\else
      \def\first@char{}%
    \fi
    \macc@nested@a\relax111{\first@char}%
  \fi
  \endgroup
}

\numberwithin{equation}{section}

\title{Gauge transformations of spectral triples with twisted real structures}
\date{September 24, 2020}
\author{Adam M. Magee and Ludwik D\k{a}browski}
\begin{document}
\maketitle

\begin{abstract}
Twisted real structures are well-motivated as a way to implement the conformal transformation of a Dirac operator for a real spectral triple without needing to twist the noncommutative 1-forms. We study the coupling of spectral triples with twisted real structures to gauge fields, adopting Morita equivalence via modules and bimodules as a guiding principle and paying special attention to modifications to the inner fluctuations of the Dirac operator. In particular, we analyse the twisted first-order condition as a possible alternative to abandoning the first-order condition in order to go beyond the Standard Model, and elaborate upon the special case of gauge transformations accordingly. Applying the formalism to a toy model, we argue that under certain physically-motivated assumptions the spectral triple based on the left-right symmetric algebra should reduce to that of the Standard Model of fundamental particles and interactions, as in the untwisted case.
\end{abstract}

\section{Introduction}

Connes' noncommutative geometry (see \eg, \cite{C94}) proposes a `quantum' generalisation of Riemannian manifolds, which are in this framework described by so-called spectral triples $(\alg{A},\hilb,D)$. Spectral triples take inspiration from the Gelfand-Naimark duality between topological spaces and commutative C*-algebras, but extend the concept to include noncommutative $*$-algebras $\alg{A}$ and further incorporate metric and calculus through the inclusion of a generalised Dirac operator $D$. By $\alg{A}$ and $D$ acting on a Hilbert space $\hilb$, the topological and geometric notions of Riemannian geometry can be translated into operator algebraic language. Real spectral triples, which are additionally equipped with a real structure $J$, appear well-suited to providing a natural mathematical framework for expressing certain gauge theories, including that of the Standard Model of particle physics, which accurately describes all presently known high energy physics phenomena. 

A detailed treatment covering the noncommutative geometric formulation of the Standard Model is provided in \cite{CMar08}, but for our purposes it is sufficient to say that the real even spectral triple describing the Standard Model comes from the product of the manifold spectral triple 
\[ (C^\infty(\mathcal{M}),L^2(\mathcal{M},\mathcal{S}),\slashed{D},J_{\mathcal{M}},\gamma_{\mathcal{M}}), \]
which describes the spatial degrees of freedom, with the finite spectral triple
\[ (\alg{A}_{\mathrm{SM}},\mathbb{C}^{96},D_{\mathrm{SM}},J_F,\gamma_F), \]
which describes the internal degrees of freedom of the theory. Here, for a manifold $\mathcal{M}$, $C^\infty(\mathcal{M})$ is the algebra of smooth complex functions on $\mathcal{M}$, $L^2(\mathcal{M},\mathcal{S})$ is the space of square-integrable spinors on $\mathcal{M}$, $\slashed{D}$ is the Dirac operator associated to the spinor bundle $\mathcal{S}$, $\alg{A}_{\mathrm{SM}}$ is the real $*$-algebra $\CC\oplus\HH\oplus M_3(\CC)$ where $\HH$ denotes the quaternions, $D_{\mathrm{SM}}$ is the fermionic mass matrix, $J_{\mathcal{M}}\otimes J_F$ is the charge conjugation operator, and $\gamma_{\mathcal{M}}\otimes\gamma_F$ is the chirality operator.

Real spectral triples have a natural notion of dimension coming from K-theoretic concepts known as \emph{KO-dimension} (see \eg, \cite{C95}), which coincides with the dimension of the manifold (modulo 8) in the commutative case. It is well-known that the KO-dimension of the finite part of the spectral triple for the Standard Model must be 6 \cite{B07,C06}. Assuming that the Hilbert space admits a symplectic structure, the smallest irreducible representation of a matrix algebra on finite-dimensional Hilbert space, whose grading is compatible with the grading on the algebra, and which is of KO-dimension 6, is $M_2(\mathbb{H})\oplus M_4(\mathbb{C})$ \cite{CC08}.\footnote{A more principled justification for considering algebras of the form $M_k(\HH)\oplus M_{2k}(\CC)$ for applications to quantum physics is offered in \cite{CCM14}.} The form of the grading breaks this algebra down to the `left-right symmetric algebra'\footnote{The reader should beware that this name is sometimes given to the algebra ${\CC\oplus\HH_L\oplus\HH_R\oplus M_3(\CC)}$.} $\alg{A}_{\mathrm{LR}}=\HH_L\oplus\HH_R\oplus M_4(\CC)$, and the fulfilment of the first-order condition, which ensures that $D$ is a first-order differential operator in noncommutative-geometric terms, breaks $\ALR$ down to $\ASM$.

As the gauge group of a spectral triple comes from the choice of finite algebra, and extensions of the Standard Model coming from enlargements of the gauge group are of ongoing physical interest, it is natural to ask if the first-order condition can be jettisoned such that $\ALR$ can be taken as the algebra of the spectral triple, and if so, what gauge theory this spectral triple would correspond to. The first of these questions was answered in the affirmative by \cite{CCvS13a}. The second question was answered by \cite{CCvS13b} and \cite{CCvS15}, and it was found that this spectral triple corresponds to a family of Pati-Salam $\SU(2)_L\times\SU(2)_R\times\SU(4)$ models.

One should not be too hasty in discarding the first-order condition, though. For one thing, it is the noncommutative equivalent of the requirement that a generalised Dirac operator be a first-order differential operator. Furthermore, it was introduced (along with the real structure $J$ which implements it) in \cite{C95} at least partly to better define the notion of gauge theories in noncommutative geometry. It would seem advantageous, then, to search for a less radical solution. A generalised notion of real structure is given by the twisted real structure of \cite{BCDS16}, which was there demonstrated to be applicable to the case of certain conformal transformations of real spectral triples. The range of applicability was subsequently further extended by \cite{DS21} using ``multitwisted'' real structures. We hence investigate the possibility that such twisted real structures might offer a route to implementing the left-right symmetric spectral triple with only a weakening, rather than a complete discarding, of the first-order condition, or if the reduction to the Standard Model is unavoidable, as occurs when imposing the (untwisted) first-order condition.

In order to do so, in \S3 (culminating in Thm. \ref{env:innerflucthm}) we present in great detail a construction of Morita (self-)equivalence bimodules for spectral triples with twisted real structure that gives the expected form of inner fluctuations of the Dirac operator (\cf \cite[\S2.2]{BCDS16}). In \S4 we use this construction to develop a notion of gauge transformations for spectral triples with twisted real structure, given in Thm. \ref{env:DwuDuw}. The necessary alterations to the spectral action are then described in \S5 and finally, in \S6 we attempt to apply the formalism to the toy model based on the algebra $\CC_L\oplus\CC_R\oplus M_2(\CC)$, which takes the role of a simplified version of the spectral Pati-Salam model.\footnote{To be precise, in \cite{CCvS13a} it is argued that this algebra gives rise to a $\U(1)_L\times\U(1)_R\times\U(2)$ gauge theory in the absence of the first-order condition.} In the course of doing so we discuss various issues and limitations we encounter both for the toy model and for the full physical model. 

\section{Twisted real structures}

Let $\alg{A}$ be a unital $*$-algebra and $\hilb$ a Hilbert space admitting a \mbox{$*$-representation} $\pi\colon\alg{A}\to\bdd$ where $\bdd$ denotes the C*-algebra of bounded linear operators on $\hilb$. 

Furthermore, let $D$ be a densely-defined self-adjoint linear operator on $\hilb$ with compact resolvent such that $\comm{D}{\pi(a)}\in\bdd$ for all $a\in\alg{A}$ (the ``generalised Dirac operator''), $J$ an antilinear map $J\colon\hilb\to\hilb$ such that $J^2=\pm1$ and $J^*=J^{-1}$ (the ``real structure''), and $\gamma\colon\hilb\to\hilb$ a $\mathbb{Z}_2$-grading operator on a Hilbert space $\hilb$ such that $\gamma^2=1$ and $\gamma^*=\gamma$. 

Let $\nu$ be a bounded operator on $\hilb$ with bounded inverse such that there exists an algebra automorphism $\hat{\nu}\colon\mathcal{A}\to\mathcal{A}$ implemented by 
\begin{equation}
\pi(\hat{\nu}(a))\coloneqq\nu\pi(a)\nu^{-1}
\end{equation}
for all $a\in\alg{A}$. We will call such a $\nu$ a \emph{twist operator}. It is also possible to define another algebra automorphism $\tilde{\nu}\colon\alg{A}\to\alg{A}$ using $\nu$ which is given by \begin{equation}\label{eq:tildetwist}\pi(\tilde{\nu}(a))\coloneqq\nu^*\pi(a)(\nu^*)^{-1}.\end{equation} Note that we can express $\tilde{\nu}$ in terms of $\hat{\nu}$ by $\tilde{\nu}(a)=(\hat{\nu}^{-1}(a^*))^*$ for any $a\in\alg{A}$.

\begin{defn}\label{env:optwistdef}
A \emph{spectral triple with twisted real structure} is a collection of spectral data $(\alg{A},\hilb,D,J,\nu)$ such that the following conditions are satisfied for all $a,b\in\alg{A}$: 
\begin{equation}\label{eq:0C}
\comm{\pi(a)}{J\pi(b)J^{-1}}=0,
\end{equation}
\begin{equation}\label{eq:nu1C}
\comm{D}{\pi(a)}J\pi(\hat{\nu}^2(b))J^{-1}=J\pi(b)J^{-1}\comm{D}{\pi(a)},
\end{equation}
\begin{equation}\label{eq:nueC}
DJ\nu=\varepsilon'\nu JD,
\end{equation}
\begin{equation}\label{eq:RC}
\nu J\nu=J,
\end{equation}
where $\varepsilon'\in\{1,-1\}$. These conditions we will refer to as, respectively, the zeroth-order condition, the $\nu$-twisted first-order condition, the $\nu$-twisted $\varepsilon'$ condition and the regularity condition.
\end{defn}

\begin{rmk}
In Def. \ref{env:optwistdef}, several technical requirements must be met for all of the expressions given to be well-defined. We do not list them explicitly, as they are `natural' and regard the preservation of the domain of the operator $D$ under the action of the various bounded operators ($\pi(a)$ for $a\in\alg{A}$, $J$, $\nu$, \etc). 
\end{rmk}

Because $\hat{\nu}$ is an automorphism of $\alg{A}$, we can equivalently express \eqn{nu1C} in the `balanced' form
\begin{equation}\label{eq:nu1Cbal} \comm{D}{\pi(a)}J\pi(\hat{\nu}(b))J^{-1}=J\pi(\hat{\nu}^{-1}(b))J^{-1}\comm{D}{\pi(a)}. \end{equation}
Compatibility of \eqn{nu1C} with the $*$-structure forces one to also require
\begin{equation}\label{eq:nu1Cstar} \comm{D}{\pi(a)}J\pi(\tilde{\nu}(b))J^{-1}=J\pi(\tilde{\nu}^{-1}(b))J^{-1}\comm{D}{\pi(a)}, \end{equation}
which as $\tilde{\nu}$ is also an automorphism, can also be expressed in the `unbalanced' form of \eqn{nu1C}. We note that \eqn{nu1Cstar} follows as a consequence of \eqn{nu1C} (rather than being taken as an additional assumption) when $\nu=\nu^*$. 

It is clear from the above definition that the familiar, `ordinary' real spectral triples $(\alg{A},\hilb,D,J)$ can be considered as special cases of spectral triples with twisted real structure with the trivial twist operator $\nu=1$. In this document, we will refer to these spectral triples as `trivially-twisted'. For the sake of ease, in this document we will retain the familiar notation omitting $\nu$, rather than writing the full $(\alg{A},\hilb,D,J,1)$ for these spectral triples.

\begin{note}
The conditions \eqn{nu1C} and \eqn{nueC} have better-known analogues coming from the trivially-twisted case. These are respectively
\begin{equation}\label{eq:1C}
\comm{D}{\pi(a)}J\pi(b)J^{-1}=J\pi(b)J^{-1}\comm{D}{\pi(a)},
\end{equation}
the first-order condition, and
\begin{equation}\label{eq:eC}
DJ=\varepsilon'JD,
\end{equation}
where $\varepsilon'\in\{-1,+1\}$, which we refer to as the $\varepsilon'$ condition. 
\end{note}

\begin{defn}\label{env:nueeC}
A spectral triple with twisted real structure $(\alg{A},\hilb,D,J,\nu)$ is called \emph{even} if its spectral data includes a grading operator $\gamma$ such that ${\gamma D=-D\gamma}$ and $\comm{\gamma}{\pi(a)}=0$ for all $a\in\alg{A}$. Furthermore,  $\gamma$ is also required to satisfy
\begin{equation}\label{eq:gnuJ}
\gamma\nu J=\varepsilon''\nu J\gamma,
\end{equation}
which we refer to as the $\nu$-twisted $\varepsilon''$ condition, with $\varepsilon''\in\{-1,+1\}$, and the commutation relation
\begin{equation}
\comm{\gamma}{\nu^2}=0.
\end{equation}
\end{defn}
\begin{note}
The condition \eqn{gnuJ} has a better-known analogue from the trivially-twisted case,
\begin{equation}\label{eq:gJ}
\gamma J=\varepsilon''J\gamma,
\end{equation} where $\varepsilon''\in\{-1,+1\}$, which we refer to as the $\varepsilon''$ condition.
\end{note}
\begin{rmk}
It is worth noting that previous papers on twisted real structures (\eg, \cite{BCDS16,BDS19}) take \eqn{gJ} rather than \eqn{gnuJ}. The motivation for using \eqn{gnuJ} in this paper instead comes from Prop. \ref{env:gradprop}, and the fact that \eqn{gnuJ} is the weaker choice of the two constraints.
\end{rmk}

Spectral triples which are not even are referred to as \emph{odd}. As a point of notation, such even spectral triples with twisted real structure will be denoted by $(\alg{A},\hilb,D,J,\gamma,\nu)$. To be consistent with the convention established earlier, trivially-twisted real even spectral triples will be denoted by $(\alg{A},\hilb,D,J,\gamma)$. Where not defined explicitly, context should make clear whether a collection of spectral data is an odd spectral triple with twisted real structure or an even trivially-twisted real spectral triple.

Spectral triples with twisted real structure were first introduced in \cite{BCDS16}. It is important to note that spectral triples with twisted real structure are different to the ``twisted spectral triples'' introduced in \cite{CM08}, and one should take special care not confuse the two. However, it is possible in many cases to draw equivalences between the two frameworks (see \cite{BDS19}), and some looser parallels can be drawn too, as we shall see in the next section. 

It will frequently occur that we will take the twist operator to be self-adjoint or involutive up to sign. When it is necessary to keep track of these signs, we will consistently refer to them as $\alpha_1$ and $\alpha_2$ respectively, \ie, if $\nu$ is self-adjoint up to sign we will say $\nu=\alpha_1\nu^*$ and if $\nu$ is involutive up to sign, we will say that $\nu=\alpha_2\nu^{-1}$. Regarding the latter case in particular, we will refer to spectral triples with such involutive-twisted real structures as being \emph{mildly-twisted}, because these spectral triples satisfy the ordinary first-order condition \eqn{1C}.

As a further point of notation, we additionally define the conjugate $*$-representation $\pi_J(a)\coloneqq J\pi(a)J^{-1}$ for any $a\in\alg{A}$. It will also be useful to further define the $*$-antirepresentation $\pi_J^*(a)\coloneqq J\pi(a)^*J^{-1}$ for any $a\in\alg{A}$. These will make it more convenient to use the following notation for twisted commutators: $$\comm{T}{a}^\pi_\sigma\coloneqq T\pi(a)-\pi(\sigma(a))T,$$ for $a\in\alg{A}$, $T$ an operator on $\hilb$, $\pi\colon\alg{A}\to\bdd$ and $\sigma$ an algebra automorphism of $\alg{A}$. Of course, in this notation we might write ordinary commutators of operators with algebra elements as $\comm{T}{a}^\pi_{\id_{\alg{A}}}$ although in this document we will favour the standard notation $\comm{T}{\pi(a)}$ for simplicity. With this notation, we can make use of the following results from \cite{BDS19}.
\begin{lem}\label{env:altnu1C}
The $\nu$-twisted first-order condition \eqn{nu1C} can be equivalently written as $$\comm{\comm{D}{a}}{b}^{\pi_J}_{\hat{\nu}^{-2}}=0$$ for any $a,b\in\alg{A}$.\end{lem}
\begin{lem}\label{env:Jacobi}
Let $T$ be an operator on $\hilb$, with the algebra $\alg{A}$ represented on $\hilb$ by $\pi\colon\alg{A}\to\bdd$ and $\pi_J\colon\alg{A}\to\bdd$, with the algebra automorphisms $\sigma,\rho\in\Aut(\alg{A})$. Then if \eqn{0C} holds, we have $$\comm*{\comm*{T}{a}^\pi_\sigma}{b}^{\pi_J}_\rho=\comm*{\comm*{T}{b}^{\pi_J}_\rho}{a}^\pi_\sigma$$  for all $a,b\in\alg{A}$.
\end{lem}

Another result that will be of some relevance later is the following proposition.
\begin{prop}\label{env:ruinerprop}
Suppose that $(\alg{A},\hilb,D,J,\gamma,\nu)$ is an even spectral triple with twisted real structure with KO-dimension signs $(\varepsilon,\varepsilon',\varepsilon'')$ whose twist operator satisfies $\nu=\alpha\nu^*=\alpha\nu^{-1}$ for $\alpha\in\{+1,-1\}$. Then $(\alg{A},\hilb,D,\nu J,\gamma)$ is an even (trivially-twisted) real spectral triple with KO-dimension signs $(\varepsilon,\alpha\varepsilon',\varepsilon'')$ provided that $\nu$ is a linear operator.

Alternatively, suppose that $(\alg{A},\hilb,D,J,\gamma,\nu)$ is an even spectral triple with twisted real structure with KO-dimension signs $(\varepsilon,\varepsilon',\varepsilon'')$ whose twist operator satisfies $\nu=\alpha_2\nu^{-1}$ for $\alpha_2\in\{+1,-1\}$ and $\nu D=\beta_1D\nu$, $\nu\gamma=\beta_2\gamma\nu$ for $\beta_i\in\{+1,-1\}$. Then this spectral triple with twisted real structure is equivalent to the even (trivially-twisted) real spectral triple $(\alg{A},\hilb,D,J,\gamma)$ with KO-dimension signs $(\varepsilon,\alpha_2\beta_1\varepsilon',\beta_2\varepsilon'')$.
\end{prop}
\begin{proof}
We begin with the first claim. Let us call $\nu J\eqqcolon\mathcal{J}$ and require that $\nu=\alpha_1\nu^*=\alpha_2\nu^{-1}$ for $\alpha_1,\alpha_2\in\{+1,-1\}$. We first check that $\mathcal{J}$ is a valid real structure. First of all, as it is the product of a linear and an antilinear operator, it is itself antilinear. It is straightforward to find  $$\mathcal{J}^*=J^*\nu^*=\alpha_1J^{-1}\nu=\alpha_1\alpha_2J^{-1}\nu^{-1}=\alpha_1\alpha_2\mathcal{J}^{-1}$$ and so it is antiunitary provided $\alpha_1\alpha_2=1$. This implies that $\alpha_1=\alpha_2$, and so we call $\alpha=\alpha_1=\alpha_2$. Furthermore, $$\mathcal{J}^2=\nu J\nu J=J^2=\varepsilon1$$ by \eqn{RC}. 

It is equally straightforward to see that \begin{align*}DJ\nu&=\varepsilon'\nu JD\\
\alpha D\nu J&=\varepsilon'\nu JD\\
D\mathcal{J}&=\alpha\varepsilon'\mathcal{J}D
\end{align*}
and so \eqn{nueC} reduces to \eqn{eC} with the sign $\alpha\varepsilon'$. Similar reasoning shows that \eqn{nu1C} reduces to \eqn{1C}, particularly if one considers the `balanced' form \begin{align*}
\comm{D}{\pi(a)}J\nu\pi(b)\nu^{-1}J^{-1}&=J\nu^{-1}\pi(b)\nu J^{-1}\comm{D}{\pi(a)}\\
\alpha^2\comm{D}{\pi(a)}\mathcal{J}\pi(b)\mathcal{J}^{-1}&=\mathcal{J}\pi(b)\mathcal{J}^{-1}\comm{D}{\pi(a)}\\
\comm{D}{\pi(a)}\mathcal{J}\pi(b)\mathcal{J}^{-1}&=\mathcal{J}\pi(b)\mathcal{J}^{-1}\comm{D}{\pi(a)}.
\end{align*}

In the even case, all of the above applies unchanged. The requirement that $\comm{\gamma}{\nu^2}=0$ is satisfied trivially as $\nu^2\propto1$ and $\gamma\nu J=\varepsilon''\nu J\gamma$ immediately becomes $\gamma\mathcal{J}=\varepsilon''\mathcal{J}\gamma$ by the definition of $\mathcal{J}$, and so \eqn{gnuJ} reduces to \eqn{gJ}.

For the second claim, $\nu^2\propto1$ immediately reduces \eqn{nu1C} to \eqn{1C}. The reduction of \eqn{nueC} to \eqn{eC} and \eqn{gnuJ} to \eqn{gJ} comes quite immediately from the fact that \eqn{RC} reduces to $\nu J=\alpha_2J\nu$ combined with the commutation relations between $\nu$ and $D$ and $\gamma$ respectively.
\end{proof}
\begin{rmk}
The alteration of the KO-dimension signs in going from a spectral triple with twisted real structure to a real spectral triple as outlined above may in some cases require a redrawing the standard table of KO-dimensions (as must be done when changing the real structure of an even real spectral triple from $J$ to $\gamma J$). However, we will not address this issue in detail in this paper, but rather assume for simplicity that we are only looking at cases where this is not an issue.
\end{rmk}

\section{Morita equivalence of spectral triples with twisted real structure}

Before we can talk about the applications of spectral triples with twisted real structure to gauge theories, we should understand how the changes to the usual definitions discussed in the previous section affect the definition of gauge transformations, and in order to do this we must discuss how the notion of Morita equivalence has been changed.

\subsection{Inner fluctuations}

The space of noncommutative 1-forms associated to a spectral triple $(\alg{A},\hilb,D)$ is generated by (the representation of) the algebra $\alg{A}$ and the derivation $\comm{D}{\place}$, and is denoted by 
\[ \Omega^1_D(\alg{A})\coloneqq\bigg\{\sum_i\pi(a_i)\comm{D}{\pi(b_i)}:a_i,b_i\in\alg{A}\bigg\}. \]
For what follows, we would like to maintain this conventional 1-form structure as much as possible. 

If the spectral triple is equipped with a (trivially-twisted) real structure $J$, an \emph{inner fluctuation} of the Dirac operator is given by 
\begin{equation}\label{eq:Dalpha}D_\omega=D+\omega+\varepsilon'J\omega J^{-1}\end{equation} for $\omega^*=\omega\in\Omega^1_D(\alg{A})$ a self-adjoint 1-form. However, if $D_\omega$ is to satisfy \eqn{nueC}, it should instead be of the form
\begin{equation}\label{eq:Dalphatwist}D_\omega=D+\omega+\varepsilon'\nu J\omega J^{-1}\nu.\end{equation} The relevant question is then whether or not we can sensibly implement such a fluctuation as \eqn{Dalphatwist}.

\begin{rmk}
Equation \eqn{Dalphatwist} is not the only possible inner fluctuation of $D$ which satisfies \eqn{nueC}. One could also define $$D_{\omega}'=D+\nu\omega\nu+\varepsilon'J\omega J^{-1},$$ now with $\nu\omega\nu\in\Omega_D^1(\alg{A})$. The difficulty with this choice is that it would necessitate modifying the structure of noncommutative 1-forms, which we would like to avoid, but instead requiring $\nu\comm{D}{a}\nu=\comm{D}{a'}$ for $a,a'\in\mathcal{A}$ then places constraints on the Dirac operator and the twist. 
\end{rmk}

\subsection{Morita equivalence}

Loosely speaking, two C*-algebras are Morita equivalent\footnote{In the context of C*-algebras, one generally refers to \emph{strong} Morita equivalence, but it is customary to omit `strong', especially when considering unital C*-algebras for which strong Morita equivalence coincides with algebraic Morita equivalence.} when there exists a categorical equivalence of their (left or right) module structures. Thus, when one considers a C*-algebra as Morita equivalent to itself, the algebra and associated Hilbert module are only `the same' up to isomorphism, and so the action of the Dirac operator on the Hilbert space is only defined up to a connection 1-form. This is the ultimate source of inner fluctuations.

Furthermore, when considering real spectral triples, the left and right module structures are related by the real structure. Therefore, finding an inner fluctuation of the Dirac operator which is compatible with the real structure is a two-step process which involves taking Morita self-equivalences of the left and the right module structures, and imposing self-consistency.

As we will demonstrate in this subsection, this procedure is more complicated in the case of twisted real structures, as the connection will be an ordinary 1-form for the right module case and a twisted 1-form for the left module case. We will largely follow \cite{L97} and \cite{LM18} for the right and left module cases respectively, but the modifications necessary to combine the two approaches and adapt them to the twisted real structure formalism are original.

\begin{defn}Two C*-algebras $\mathcal{A}$ and $\mathcal{B}$ are \emph{Morita equivalent} if there exists a Hilbert bimodule $\bimod{\mathcal{E}}{\mathcal{B}}{\mathcal{A}}$ such that
\begin{enumerate}
\item It is full as a bimodule, \ie, $\linspan{\{\mip{e_1}{e_2}_{\mathcal{A}},e_1,e_2\in\bimod{\mathcal{E}}{\mathcal{B}}{\mathcal{A}}\}}$ is dense in $\mathcal{A}$ and $\linspan{\{\lmip{e_1}{e_2}{\mathcal{B}},e_1,e_2\in\bimod{\mathcal{E}}{\mathcal{B}}{\mathcal{A}}\}}$ is dense in $\mathcal{B}$;
\item $\lmip{e_1}{e_2}{\mathcal{B}}e_3=e_1\mip{e_2}{e_3}_{\mathcal{A}}$ for all $e_1,e_2,e_3\in\bimod{\mathcal{E}}{\mathcal{B}}{\mathcal{A}}$.
\end{enumerate}
We will be primarily interested in the case of \emph{Morita self-equivalence}, where we take $\mathcal{B}=\mathcal{A}$ and consider $\bimod{\mathcal{E}}{\mathcal{B}}{\mathcal{A}}$ as the bimodule of $\mathcal{A}$ over itself.\end{defn}

A standard result which we will later make use of is that if $\mathcal{E}$ is a full and finitely-generated and projective (hereafter, `finite projective') (left or right) Hilbert module over the \mbox{C*-algebra} $\mathcal{A}$, then $\mathcal{A}$ is Morita equivalent to $\End{\mathcal{A}}{\mathcal{E}}$. 

If two C*-algebras are Morita equivalent, they can be said to have equivalent representation theories. To see this, suppose two C*-algebras $\mathcal{A}$ and $\mathcal{B}$ are Morita equivalent via the bimodule $\bimod{\mathcal{E}}{\mathcal{B}}{\mathcal{A}}$, with $\mathcal{A}$ represented on the Hilbert space $\mathcal{H}$ as bounded operators by the map $\pi_{\mathcal{A}}$. This allows us to define the new Hilbert space 
\[ \mathcal{H}'\coloneqq\bimod{\mathcal{E}}{\mathcal{B}}{\mathcal{A}}\otimes_{\mathcal{A}}\mathcal{H} \] 
such that 
\[ ea\otimes_{\mathcal{A}}\psi=e\otimes_{\mathcal{A}}\pi_{\mathcal{A}}(a)\psi \] 
for all $e\in\bimod{\mathcal{E}}{\mathcal{B}}{\mathcal{A}}$, $a\in\mathcal{A}$ and $\psi\in\mathcal{H}$, equipped with the inner product 
\[ \mip{e_1\otimes_{\mathcal{A}}\psi_1}{e_2\otimes_{\mathcal{A}}\psi_2}_{\mathcal{H}'}\coloneqq\mip{\psi_1}{\mip{e_1}{e_2}_{\mathcal{A}}\psi_2}_{\mathcal{H}} \] 
for all $e_1,e_2\in\bimod{\mathcal{E}}{\mathcal{B}}{\mathcal{A}}$ and $\psi_1,\psi_2\in\mathcal{H}$. One can then construct a representation of $\mathcal{B}$ on $\mathcal{H}'$ by 
\[ \pi_\alg{B}(b)(e\otimes_{\mathcal{A}}\psi)\coloneqq(be)\otimes_{\mathcal{A}}\psi \] 
for all $b\in\mathcal{B}$ and $e\otimes_{\mathcal{A}}\psi\in\mathcal{H}'$. One then finds that the two representations $(\mathcal{A},\pi_{\mathcal{A}},\mathcal{H})$ and $(\mathcal{B},\pi_\alg{B},\mathcal{H}')$ are equivalent. 

Since the above construction produces the representation $(\mathcal{B},\pi_\alg{B},\mathcal{H}')$ using the fact that $\bimod{\mathcal{E}}{\mathcal{B}}{\mathcal{A}}$ is a right $\mathcal{A}$-module, we refer to it as Morita equivalence by right module. Of course, if $\mathcal{A}$ and $\mathcal{B}$ are Morita equivalent, then there also exists an $\mathcal{A}$-$\mathcal{B}$-bimodule, the conjugate bimodule $\bimod{\bar{\mathcal{E}}}{\mathcal{A}}{\mathcal{B}}$, which allows one to start with a representation of $\mathcal{B}$ and construct a unitarily equivalent representation of $\mathcal{A}$ using the fact that the conjugate bimodule is a left $\mathcal{A}$-module. We refer to this construction as Morita equivalence by left module.

\subsubsection{Morita self-equivalence by right module}

For this part, we will not need the real structure, and the twist will play no role, so we simply summarise the standard construction. Consider a spectral triple $(\alg{A},\pi,\hilb,D)$, and a representation $(\alg{B},\pi_{\alg{B}},\hilb_R)$ equivalent to $(\alg{A},\pi,\hilb)$ by Morita \emph{self}-equivalence by right module. We begin by considering the simplest way to construct a Dirac operator on $\hilb_R$, which is by simply taking the na\"{i}ve action of $D$ given by $D_r(a\otimes\psi)\coloneqq a\otimes D\psi$. However, this does not respect the module structure as $D_r(e\otimes\pi(a)\psi)\neq D_r(ea\otimes\psi)$ for all $a\in\alg{A},e\in\bimod{\alg{E}}{\alg{B}}{\alg{A}}$ since $D$ is not assumed to commute with $\pi(\alg{A})$. Instead we find
\begin{align*}
D_r(ea\otimes\psi)&=ea\otimes D\psi\\
&=e\otimes\pi(a)D\psi\\
&=e\otimes(D\pi(a)-\comm{D}{\pi(a)})\psi\\
&=D_r(e\otimes\pi(a)\psi)-e\otimes\comm{D}{\pi(a)}\psi.
\end{align*}

Therefore we instead consider a connection on the bimodule 
\[ \nabla\colon\EBA\to\EBA\otimes_{\alg{A}}\Omega^1(\alg{A}) \]
satisfying the Leibniz rule
\[ \nabla(ea)=(\nabla(e))a+e\otimes\delta(a), \]
where $\mathrm{Der}(\alg{A})\ni\delta\colon\alg{A}\to\Omega^1(\alg{A})$. When $\EBA$ is finite projective \ie, $\EBA=p\alg{A}^N$ for $p=p^*=p^2\in M_N(\alg{A})$, we can define any such connection as the Grassmann connection $p\delta$ up to some $\omega=p\omega=\omega p=p\omega p\in\End{\alg{A}}{\EBA,\EBA\otimes_{\alg{A}}\Omega^1(\alg{A})}$.

We now impose the Morita self-equivalence by taking $\mathcal{B}=\alg{A}$ and treating the Morita equivalence bimodule $\mathcal{E}=p\mathcal{A}^N$ as the module of $\alg{A}$ over itself (\ie, take $p=1$ and $N=1$). We have that the Hilbert spaces are isomorphic $\hilb_R=\alg{A}\otimes_{\alg{A}}\hilb\overset{\iota}{\simeq}\hilb$ under the isomorphism
\[ \iota\colon a\otimes\psi\mapsto\pi(a)\psi \] 
with inverse $\iota^{-1}\colon\psi\mapsto 1\otimes\psi$. Then the representation of the algebra $\alg{B}=\alg{A}$ is simply given by 
\[ \iota\circ\pi_{\alg{B}}(b)\circ\iota^{-1}=\pi(b) \]
for all $b\in\alg{A}$. Working in the Dirac calculus, the derivation $\delta$ is given by $$\delta(a)\coloneqq\comm{D}{\pi(a)}$$ for all $a\in\alg{A}$, which generates the space of noncommutative 1-forms $\Omega^1_D(\alg{A})$, which is an $\alg{A}$-bimodule with bimodule product $$a\cdot\omega\cdot a'\coloneqq\pi(a)\omega\pi(a')$$ for $a,a'\in\alg{A}$ and $\omega\in\Omega^1_D(\alg{A})$. We then have that any $\Omega^1_D(\alg{A})$-valued connection on the right module $\alg{A}$ reads\footnote{We understand $\omega\colon\alg{A}\to\Omega^1_D(\alg{A})$ to act via the module product, \ie, $\omega(a)=\omega\cdot a=\omega\pi(a)$ as an operator.} $$\nabla=\delta+\omega$$ for $\omega\in\Omega^1_D(\alg{A})$. 

This suggests that the appropriate construction for the Dirac operator comes from making $D_r$ compatible with the module structure by the addition of a connection like so:
\begin{align*}
D_R(a\otimes\psi)\coloneqq{}&a\otimes D\psi+\nabla(a)\psi\\
={}&a\otimes D\psi+1\otimes\delta(a)\psi+1\otimes(\omega\cdot a)\psi\\
={}&1\otimes D\pi(a)\psi+1\otimes\omega\pi(a)\psi\\
={}&D_R(1\otimes\pi(a)\psi),
\end{align*} where the last line comes from observing that $\nabla(1)=\delta(1)+(\omega\cdot1)=\omega$ as an operator. Lastly, via the isomorphism $\iota$ we find the compatible Dirac operator $$D_R=D+\omega$$ on $\hilb$ with $\omega\in\Omega^1_D(\alg{A})$. 

Thus we took the representation $(\alg{A},\pi,\hilb)$ with Dirac operator $D$ on $\hilb$ and found the Morita self-equivalent representation $(\alg{A},\pi,\hilb)$ with Dirac operator $D_R=D+\omega$ on $\hilb$ for $\omega\in\Omega^1_D(\alg{A})$. However, if our original spectral triple is a real spectral triple with real structure $J$, then along with the twist operator $\nu$ it necessarily obeys \eqn{nueC}. The same cannot be said for $(\alg{A},\hilb,D+\omega,J,\nu)$, which obeys \eqn{nueC} if and only if $\omega=\varepsilon'\nu J\omega J^{-1}\nu$. But it is not difficult to show that, for $\omega=\sum_j\pi(a_j)\comm{D}{\pi(b_j)}\in\Omega^1_D(\alg{A})$, 
\begin{align}\varepsilon'\nu J\omega J^{-1}\nu&=\sum_j\nu J\pi(a_j)J^{-1}\nu^{-1}(D\nu^{-1}J\pi(b_j)J^{-1}\nu-\nu J\pi(b_j)J^{-1}\nu^{-1}D)\nonumber\\
&=\sum_j\pi_J(\hat{\nu}^{-1}(a_j))\comm{D}{\hat{\nu}(b_j)}^{\pi_J}_{\hat{\nu}^{-2}},\label{eq:nuJomegaJnu}\end{align} 
which has no reason to be equal to $\omega$.

\subsubsection{Morita self-equivalence by left module}

The significance of \eqn{nuJomegaJnu} is that if $\omega$ is a 1-form, then $\nu J\omega J^{-1}\nu$ is a twisted 1-form. We have already seen that the standard Morita self-equivalence by right module can obtain $\omega$, so we expect that some changes should need to be made to the construction in the left module case to obtain the twisted 1-form $\nu J\omega J^{-1}\nu$. For said changes, we look to the construction offered by \cite{LM18}  for obtaining inner fluctuations of the Dirac operator for real twisted spectral triples in the spirit of \cite[Prop. 3.4]{CM08}.

As telegraphed, given an $\alg{A}$-$\alg{B}$ Morita equivalence bimodule $\EBA$, there also exists a conjugate bimodule $\bimod{\bar{\alg{E}}}{\alg{A}}{\alg{B}}$ which has the canonical product $a\bar{e}b=\widebar{b^*ea^*}$ where $\bar{\alg{E}}=\{\bar{e}:e\in\alg{E}\}$. If we replace $a_j\mapsto a_j^*$ and $b_j\mapsto b_j^*$ in \eqn{nuJomegaJnu} then we find that we can write
\[ \varepsilon'\nu J\omega J^{-1}\nu=\sum_j\pi^*_J(\tilde{\nu}(a_j))\comm{D}{\tilde{\nu}^{-1}(b_j)}^{\pi^*_J}_{\tilde{\nu}^2}. \]
Motivated by this equality, we define the space of ``twisted-opposite'' 1-forms as
\begin{equation}\label{eq:twistopspace} \tilde{\Omega}^1_D(\alg{A}^{\mathrm{op}})\coloneqq\bigg\{\sum_j\pi^*_J(\tilde{\nu}(a_j))\comm{D}{\tilde{\nu}^{-1}(b_j)}^{\pi^*_J}_{\tilde{\nu}^2}:a_j,b_j\in\alg{A}\bigg\}. \end{equation}
In light of this, it will prove convenient to define two ``twisted-opposite'' maps 
\begin{subequations}
\begin{align}
a^\oplus&\coloneqq\pi^*_J\big(\tilde{\nu}(a)\big)=J\nu^{-1}\pi(a)^*\nu J^{-1},\\
a^\ominus&\coloneqq\pi^*_J\left(\tilde{\nu}^{-1}(a)\right)=J\nu\pi(a)^*\nu^{-1}J^{-1},
\end{align} 
\end{subequations}
noting that both $a^\oplus$ and $a^\ominus$ are (still) elements of $\pi^*_J(\alg{A})\equiv\alg{A}^{\circ}$ for $a\in\alg{A}$, and that each twisted-opposite map separately preserves the algebra product. This notation allows us to write twisted-opposite 1-forms in a more compact fashion:
\begin{lem}\label{env:nuJomegaJnu}
The operator $\varepsilon'\nu J\omega J^{-1}\nu$, for $\omega=\sum_j\pi(a_j^*)\comm{D}{\pi(b_j^*)}\in\Omega^1_D(\alg{A})$, can be rewritten in the form $\sum_ja_j^\oplus(Db_j^\ominus-b_j^\oplus D)$.
\end{lem}
\begin{proof}
The proof is simply by computation. For the sake of simplicity and without loss of generality, we will omit summations and the representation $\pi$.
\begin{align*}
\varepsilon'\nu J\omega J^{-1}\nu&=\varepsilon'\nu Ja^*\comm{D}{b^*}J^{-1}\nu\\
&=\varepsilon'\nu J(a^*Db^*-a^*b^*D)J^{-1}\nu\\
&=\varepsilon'(\nu Ja^*Db^*J^{-1}\nu-\nu Ja^*b^*DJ^{-1}\nu)\\
&=\nu Ja^*\nu JDJ^{-1}\nu b^*J^{-1}\nu-\nu J^{-1}a^*b^*\nu JDJ^{-1}\nu J\nu\text{ by \eqn{nueC},}\\
&=J\nu^{-1}a^*\nu J^{-1}DJ\nu b^*\nu^{-1}J^{-1}-J\nu^{-1}a^*b^*\nu J^{-1}D\\
&=a^\oplus Db^\ominus-a^\oplus b^\oplus D.
\end{align*}
\end{proof}
Thus the space of twisted-opposite 1-forms \eqn{twistopspace} could equivalently be defined
$$\tilde{\Omega}^1_D(\alg{A}^{\mathrm{op}})=\bigg\{\sum_ja_j^\oplus(Db_j^\ominus-b_j^\oplus D):a_j^\oplus,b_j^\ominus,b_j^\oplus\in\alg{A}^{\circ}\bigg\},$$ 
which the reader may find easier to parse. We denote the elements of this space by $\omega^\odot$ to contrast with the more familiar $\omega^\circ\coloneqq\varepsilon'J\omega J^{-1}$. This space can be considered an $\alg{A}$-bimodule with bimodule action defined by $$a\cdot\omega^\odot\cdot b\coloneqq b^\oplus\omega^\odot a^\ominus$$ for all $a,b\in\alg{A}$ and $\omega^\odot\in\tilde{\Omega}^1_D(\alg{A}^{\mathrm{op}})$, and is generated by the derivation \begin{equation}\label{eq:newder}\delta^\odot(a)\coloneqq Da^\ominus-a^\oplus D\end{equation} for $a\in\alg{A}$. This is a derivation in the sense that:\footnote{Of course, rather than taking $\delta^\odot$ to be a derivation with respect to an unusual bimodule action, one could take it to be a twisted derivation with respect to the usual bimodule action, obeying the twisted Leibniz rule for $\tilde{\nu}^2$ the twist: $\delta^\odot(ab)=\delta^\odot(b)a^\ominus+\left(\tilde{\nu}(b)^2\right)^\ominus\delta^\odot(a)$. We choose to follow \cite{LM18}, which also keeps us in close contact with the typical construction of a connection on a Hilbert module.}
\begin{align*}
\delta^\odot(ab)&=D(ab)^\ominus-(ab)^\oplus D\\
&=Db^\ominus a^\ominus-b^\oplus a^\oplus D\\
&=Db^\ominus a^\ominus-b^\oplus Da^\ominus+b^\oplus Da^\ominus-b^\oplus a^\oplus D\\
&=\delta^\odot(b)a^\ominus+b^\oplus\delta^\odot(a)\\
&=a\cdot\delta^\odot(b)+\delta^\odot(a)\cdot b.
\end{align*}

We now return to the matter of Morita equivalence. Finding a Morita equivalent representation of $(\alg{A},\hilb)$ by left module (using now the conjugate module) follows similarly to the right module construction given above. We denote the resultant representation as $(\alg{B},\hilb_L)$, where \[ \hilb_L\coloneqq\hilb\otimes_{\alg{A}}\conjEBA \] is the Hilbert space with the inner product
\[ \mip{\psi_1\otimes \bar{e}_1}{\psi_2\otimes\bar{e}_2}_{\hilb_L}=\mip{\psi_1\tensor[_{\alg{A}}]{\mip{\bar{e}_1}{\bar{e}_2}}{}}{\psi_2}_\hilb \]
for all $\psi_i\in\hilb$ and $\bar{e}_i\in\conjEBA$ (as a left $\alg{A}$-module), where $\tensor[_{\alg{A}}]{\mip{\bar{e}_1}{\bar{e}_2}}{}\coloneqq\mip{e_1}{e_2}_{\alg{A}}$. Furthermore, the right action of $\alg{B}\simeq\End{\alg{A}}{\lmod{\alg{E}}{\alg{A}}}$ on $\tensor[_{\alg{A}}]{\alg{E}}{}$ is extended to $\hilb_L$ by $$(\psi\otimes e)b\coloneqq\psi\otimes eb.$$ 
In the standard construction, we would use the standard algebra bijection \begin{equation}\label{eq:op}a^\circ\coloneqq\pi_J^*(a)=J\pi(a)^*J^{-1}\end{equation} for the right action of $\alg{A}$ on $\hilb$,\footnote{This right action can be viewed as an antirepresentation of $\alg{A}$ or equivalently, a representation of the opposite algebra $\alg{A}^{\mathrm{op}}$ on $\hilb$, hence our use of $\alg{A}^{\mathrm{op}}$ in \eg, $\tilde{\Omega}^1_D(\alg{A}^{\mathrm{op}})$ and references to ``(twisted-)opposite maps'', \etc}
such that $\psi\otimes a\bar{e}=\psi a\otimes\bar{e}\coloneqq\pi^*_J(a)\psi\otimes\bar{e}$. However, this is where we instead choose to start making changes.

The first change we will make is to use the twisted-opposite maps to require that the right action of $\alg{A}$ on the original Hilbert space be now given by
\begin{equation}\label{eq:psia}
\psi a\coloneqq a^\ominus\psi=J\nu\pi(a)^*\nu^{-1}J^{-1}\psi
\end{equation} for all $a\in\alg{A}$ and $\psi\in\hilb$, which extends to the module structure in the obvious way. This choice of right action is of course not unique, but its motivation will soon become clear.

As in the right module case, the na\"{i}ve implementation of the action of the Dirac operator on $\hilb_L$ given by $$D_\ell(\psi\otimes \bar{e})\coloneqq D\psi\otimes\bar{e}$$ fails to be compatible with the module structure since $D$ does not commute with the algebra and we require the tensor product to be balanced for $\mathcal{A}$, not merely $\CC$. 

Note that, as $\tensor[_{\alg{A}}]{\alg{\bar{E}}}{}$ is finitely-generated and projective by assumption, we have $\tensor[_{\alg{A}}]{\alg{\bar{E}}}{}\simeq\alg{A}^Np$, $p=p^2=p^*\in M_N(\alg{A})$. Thus we introduce an invertible linear module map 
\begin{equation}\label{eq:oldmu} \tilde{\nu}_\alg{E}\colon\tensor[_{\alg{A}}]{\alg{\bar{E}}}{}\to\tensor[_{\alg{A}}]{\alg{\bar{E}}}{}\end{equation}
whose action is given elementwise by $\tilde{\nu}\in\Aut(\alg{A}),$ the twist automorphism of \eqn{tildetwist}, and under which we assume $p$ is invariant. This allows us to make our second change, which is to the construction of the candidate Dirac operator
$${((\id\otimes\tilde{\nu}_\alg{E}^2)\circ D_\ell)(\psi\otimes\bar{e})=D\psi\otimes\tilde{\nu}_\alg{E}^2(\bar{e})}$$ for all $\psi\in\hilb$ and $\bar{e}\in\tensor[_{\alg{A}}]{\alg{\bar{E}}}{}$. Of course, this operator is still not compatible with the tensor product:
\begin{equation}
((\id\otimes\tilde{\nu}_\alg{E}^2)\circ D_\ell)(\psi\otimes a\bar{e})-(( \id \otimes\tilde{\nu}_\alg{E}^2)\circ D_\ell)(\psi a\otimes \bar{e})=-\delta^\odot(a)\psi\otimes\tilde{\nu}_\alg{E}^2(\bar{e})\neq0.\label{eq:naive}
\end{equation}
However, just as in the right module case, the presence of a derivation suggests that we try to introduce a connection.

Thus, working from the derivation \eqn{newder}, we define an $\tilde{\Omega}^1_D(\alg{A}^{\mathrm{op}})$-valued connection on the left $\alg{A}$-module $\tensor[_{\alg{A}}]{\alg{\bar{E}}}{}$ as a map $\nabla^\odot\colon\tensor[_{\alg{A}}]{\alg{\bar{E}}}{}\to\tilde{\Omega}^1_D(\alg{A}^{\mathrm{op}})\otimes_{\alg{A}}\tensor[_{\alg{A}}]{\alg{\bar{E}}}{}$ such that $$\nabla^\odot( a\bar{e} )-a\cdot\nabla^\odot(\bar{e})=\delta^\odot(a)\otimes\bar{e}$$ for all $a\in\alg{A}$ and $\bar{e}\in\tensor[_{\alg{A}}]{\alg{\bar{E}}}{}$, with left multiplication by $\alg{A}$ on $\tilde{\Omega}^1_D(\alg{A}^{\mathrm{op}})\otimes_{\alg{A}}\tensor[_{\alg{A}}]{\alg{\bar{E}}}{}$ given by the left module structure of $\tilde{\Omega}^1_D(\alg{A}^{\mathrm{op}})$. 

Using the action of $\tilde{\Omega}^1_D(\alg{A}^{\mathrm{op}})$ on $\hilb$ we can define the map $\nabla^\odot\colon\hilb\otimes_{\mathbb{C}}\tensor[_{\alg{A}}]{\alg{\bar{E}}}{}\to\hilb\otimes_{\mathbb{C}}\tensor[_{\alg{A}}]{\alg{\bar{E}}}{}$ by
\begin{equation}\label{eq:psiconn}\nabla^\odot(\psi\otimes\bar{e})\coloneqq\psi\nabla^\odot(\bar{e})\end{equation} for all $\bar{e}\in\tensor[_{\alg{A}}]{\alg{\bar{E}}}{}$ and $\psi\in\hilb$. We cannot factorise this into a map on $\hilb\otimes_{\alg{A}}\tensor[_{\alg{A}}]{\alg{\bar{E}}}{}$ because $\psi\nabla^\odot( a\bar{e} )-(\psi a)\nabla^\odot(\bar{e})$ need not vanish. However, the obstruction is captured by the derivation $\delta^\odot$ because the actions of $\alg{A}$ and $\tilde{\Omega}^1_D(\alg{A}^{\mathrm{op}})$ are compatible, \ie, 
\begin{equation}\label{eq:Aomegacomp}(a\cdot\omega^\odot)\psi=\omega^\odot a^\ominus\psi=\omega^\odot(\psi a).\end{equation}
In short, this means that \begin{equation}\label{eq:connLeib}\psi\nabla^\odot(a\bar{e})-\psi a\nabla^\odot(\bar{e})=\delta^\odot(a)\psi\otimes\bar{e}.\end{equation}

\begin{rmk}One notes that, though of course the meanings and rules established are distinct, much of what we have stated (and will state) for objects like $\nabla^\odot$, $\omega^\odot$ and $\delta^\odot$ are analogous to equivalent statements in the familiar case with $\nabla^\circ$, $\omega^\circ$ and $\delta^\circ$ instead. Indeed, many proofs carry over analogously with only the need to substitute different symbols, though to some extent this comes as a result of deliberate choices of notation. \end{rmk}

By \eqn{psiconn} we therefore have 
\begin{align}
((\id\otimes\tilde{\nu}_\alg{E}^2)\circ\nabla^\odot)(\psi\otimes a\bar{e})-((\id\otimes\tilde{\nu}_\alg{E}^2)\circ\nabla^\odot)(\psi a\otimes\bar{e})&=( \id \otimes\tilde{\nu}_\alg{E}^2)(\delta^\odot(a)\psi\otimes\bar{e})\nonumber\\
&=\delta^\odot(a)\psi\otimes\tilde{\nu}_\alg{E}^2(\bar{e}),\label{eq:fix}
\end{align}
and so combining \eqn{naive} and \eqn{fix} we find that the correct construction for the Dirac operator on $\hilb_L$ is given by
$$D_L\coloneqq( \id \otimes\tilde{\nu}_\alg{E}^2)\circ(D_\ell+\nabla^\odot),$$
which is compatible with the module structure since $D_L(\psi\otimes a\bar{e} )-D_L(\psi a\otimes\bar{e})=0$ as desired. 

As before, we take the left $\alg{A}$-module $\tensor[_{\alg{A}}]{\alg{\bar{E}}}{}$ to be finite projective, we know that $\tensor[_{\alg{A}}]{\alg{\bar{E}}}{}\simeq\alg{A}^Np$ with $p=p^*=p^2\in M_N(\alg{A})$. In these terms, the connection decomposes as $$\nabla^\odot=\nabla^\odot_0+\vec{\omega}^\odot$$ with `twisted' Grassmann connection $$\nabla_0^\odot(\bar{e})=(\delta^\odot(e_1),\ldots,\delta^\odot(e_N))p$$ for all $\bar{e}=(e_1,\ldots,e_N)\in\tensor[_{\alg{A}}]{\alg{\bar{E}}}{}$ with $e_j\in\alg{A}$. Meanwhile, $\vec{\omega}^\odot$ is a map $\tensor[_{\alg{A}}]{\alg{\bar{E}}}{}\to\tilde{\Omega}^1_D(\alg{A}^{\mathrm{op}})\otimes_\alg{A}\tensor[_{\alg{A}}]{\alg{\bar{E}}}{}$ which is $\alg{A}$-linear in the sense that $$\vec{\omega}^\odot( a\bar{e} )=a\cdot\vec{\omega}^\odot(\bar{e}).$$ 

We now impose the self-equivalence by taking $\alg{B}=\alg{A}$ and $\conjEBA=\alg{A}$ as a module (\ie, taking $p=1$, $N=1$ considering $\conjEBA$ as finite projective over $\alg{A}$) such that $\hilb_L=\hilb\otimes_{\alg{A}}\alg{A}\simeq\hilb$.

\begin{prop}In the case of Morita self-equivalence, with $\alg{B}=\alg{A}$ and $\conjEBA=\alg{A}$, the form of the Dirac operator on $\hilb_L$ is nothing but the bounded perturbation $$D_L=D+\omega^\odot=D+\varepsilon'\nu J\omega J^{-1}\nu$$ for some $\omega^\odot=\varepsilon'\nu J\omega J^{-1}\nu\in\tilde{\Omega}^1_D(\alg{A}^{\mathrm{op}})$ with $\omega\in\Omega^1_D(\alg{A})$.\end{prop}
\begin{proof}
Because $\conjEBA\simeq\alg{A}^Np$ with $p=1$ and $N=1$, we have $\nabla^\odot\colon\alg{A}\to\tilde{\Omega}^1_D(\alg{A}^{\mathrm{op}})\otimes_{\alg{A}}\alg{A}$ such that  $\nabla^\odot=\nabla^\odot_0+\vec{\omega}^\odot$ with \begin{align*}\nabla^\odot_0(a)&=\delta^\odot(a)\otimes1,\\
\vec{\omega}^\odot(a)&=(\omega^\odot a^\ominus)\otimes1,
\end{align*} where $\omega^\odot\in\tilde{\Omega}^1_D(\alg{A}^{\mathrm{op}})$. We therefore find
\begin{align*}
D_L(\psi\otimes a)&=( \id \otimes\tilde{\nu}_\alg{E}^2)\circ(D_\ell+\nabla^\odot)(\psi\otimes a)\\
&=(( \id \otimes\tilde{\nu}_\alg{E}^2)\circ D_\ell)(\psi\otimes a)+(( \id \otimes\tilde{\nu}_\alg{E}^2)\circ \delta^\odot)(\psi\otimes a)\\
&\qquad+(( \id \otimes\tilde{\nu}_\alg{E}^2)\circ \omega^\odot)(\psi\otimes a)\\
&=D\psi\otimes\tilde{\nu}^2(a)+( \id \otimes\tilde{\nu}_\alg{E}^2)(\delta^\odot(a)\psi\otimes1)+( \id \otimes\tilde{\nu}_\alg{E}^2)(\omega^\odot a^\ominus\psi\otimes1)\\
&=(D\psi)\tilde{\nu}^2(a)\otimes1+\delta^\odot(a)\psi\otimes1+\omega^\odot a^\ominus\psi\otimes1\\
&=a^\oplus D\psi\otimes1+Da^\ominus\psi\otimes1-a^\oplus D\psi\otimes1+\omega^\odot a^\ominus\psi\otimes1\\
&=Da^\ominus\psi\otimes1+\omega^\odot a^\ominus\psi\otimes1,
\end{align*}
where we have used the fact that $(\tilde{\nu}^2(a))^\ominus=(\tilde{\nu}^{-1}(\tilde{\nu}^2(a)))^\circ=(\tilde{\nu}(a))^\circ=a^\oplus$. By making the identification $\hilb\otimes_\alg{A}\alg{A}\simeq\hilb$ via the identification of $\psi\otimes a=a^\ominus\psi\otimes1$ with $\psi$, one immediately finds that $D_L=D+\omega^\odot$. The final result then follows as a consequence of Lem.~\ref{env:nuJomegaJnu}.
\end{proof}

Something to take note of before moving on is what happens when the original spectral triple $(\alg{A},\hilb,D)$ is even. In that case, we have the following lemma.
\begin{lem}\label{env:ganticomm1}
If a grading operator $\gamma$ anticommutes with the Dirac operator $D$ and commutes with (representations of) all $a\in\alg{A}$, then $\gamma$ also anticommutes with all $\omega\in\Omega^1_D(\alg{A})$.
\end{lem}
As a result of this lemma, the grading operator will automatically anticommute with $D+\omega$ for $\omega\in\Omega^1_D(\alg{A})$. We now show that it also anticommutes with $D+\varepsilon'\nu J\omega J^{-1}\nu$ as well.
\begin{prop}\label{env:gradprop}
Let $\gamma$ be a grading operator which anticommutes with the Dirac operator $D$ and commutes with (representations of) any $a\in\alg{A}$. Then $\gamma$ also anticommutes with $D+\varepsilon'\nu J\omega J^{-1}\nu$ for any $\omega\in\Omega^1_D(\alg{A})$ provided \begin{equation}\label{eq:gnuJguess}\gamma\nu J=\varepsilon''\nu J\gamma\end{equation} for $\varepsilon''\in\{-1,+1\}$, and \begin{equation}\label{eq:gcomm}\comm{\gamma}{\nu^2}=0.\end{equation}
\end{prop}
\begin{proof}
We first focus on the second term of the fluctuated Dirac operator. By \eqn{gnuJguess}, we have that $$\gamma(\varepsilon'\nu J\omega J^{-1}\nu)=\varepsilon'\varepsilon''\nu J\gamma\omega J^{-1}\nu=-\varepsilon'\varepsilon''\nu J\omega\gamma J^{-1}\nu,$$
where the second equality is due to Lem.~\ref{env:ganticomm1}. We now have that
\begin{align*}
\gamma J^{-1}\nu&=\varepsilon\gamma J\nu\\
&=\varepsilon\gamma\nu J\nu^2\text{ using \eqn{RC},}\\
&=\varepsilon\varepsilon''\nu J\gamma\nu^2\text{ by \eqn{gnuJguess},}\\
&=\varepsilon\varepsilon''\nu J\nu^2\gamma\text{ by \eqn{gcomm},}\\
&=\varepsilon\varepsilon''J\nu\gamma\\
&=\varepsilon''J^{-1}\nu\gamma,
\end{align*} and therefore $$\gamma(\varepsilon'\nu J\omega J^{-1}\nu)=-\varepsilon'\varepsilon''\nu J\omega\gamma J^{-1}\nu=-\varepsilon'\varepsilon''\nu J\omega\varepsilon''J^{-1}\nu\gamma=-(\varepsilon'\nu J\omega J^{-1}\nu)\gamma.$$ As $\gamma D=-D\gamma$ by assumption, this is sufficient to establish the result.
\end{proof}
\begin{rmk}Indeed, that $\varepsilon'\nu J\omega J^{-1}\nu$ should anticommute with $\gamma$ is what motivates Defn. \ref{env:nueeC}. Contrast this with what is taken in the literature, $\gamma J=\varepsilon''J\gamma$, \cf \cite{BCDS16,BDS19}, which is insufficient to establish the anticommutation, even assuming $\gamma\nu^2=\nu^2\gamma$. Note that, alternatively to $\gamma\nu J=\varepsilon''\nu J\gamma$, one could instead take $\gamma J\nu=\varepsilon''J\nu\gamma$. Both only hold simultaneously if $\comm{J}{\nu}=0$, which by \eqn{RC} only happens when $\nu^2=1$.\end{rmk}

Thus we find that for a 1-form $\omega$, one has that $(\alg{A},\hilb,D+\varepsilon'\nu J\omega J^{-1}\nu,\gamma)$ is an even spectral triple. However, it fails to admit $(J,\nu)$ as a twisted real structure for more or less the same reason as in the right module case. We will resolve this problem for both left and right module cases in the next section.

\subsubsection{Bimodule and twisted real structure}

To ensure the compatibility of the Dirac operator constructed from Morita self-equivalences with the twisted real structure, one needs to combine the above two left and right module constructions. First, one fluctuates spectral triple $(\alg{A},\hilb,D)$ using the bimodule $\EBA=\alg{A}$, and then one fluctuates the resulting triple by the conjugate bimodule $\conjEBA=\alg{A}$. This yields the triple $(\alg{A},\hilb,D')$ where \begin{equation}\label{eq:D'}D'\coloneqq D+\omega_R+\varepsilon'\nu J\omega_LJ^{-1}\nu\end{equation} with $\omega_L$ and $\omega_R$ two \emph{a priori} distinct elements of $\Omega^1_D(\alg{A})$. The compatibility of \eqn{D'} with \eqn{nueC} for the real structure $J$ and twist operator $\nu$ can then always be demanded, as the following proposition guarantees. 
\begin{prop}
The Dirac operator $D'$ satisfies \eqn{nueC} if and only if there exists an element $\omega\in\Omega^1_D(\alg{A})$ such that $$D'=D+\omega+\varepsilon'\nu J\omega J^{-1}\nu.$$
\end{prop}
\begin{proof}
We have that $D'=D+\omega_R+\varepsilon'\nu J\omega_LJ^{-1}\nu$. For $D'$ to satisfy \eqn{nueC}, we must have $D'J\nu=\varepsilon'\nu JD'$, which is the case if and only if $$(\omega_L-\omega_R)-\varepsilon'\nu J(\omega_L-\omega_R)J^{-1}\nu=0.$$ Adding half of the left hand side of this equation (which is still equal to $0$) to the right hand side of \eqn{D'}, one gets $$D'=D+\frac{1}{2}(\omega_R+\omega_L)+\varepsilon'\nu J\frac{1}{2}(\omega_R+\omega_L)J^{-1}\nu.$$ This gives the claimed result for $\omega=\frac{1}{2}(\omega_R+\omega_L).$
\end{proof}

The sum total of these past three subsections can thus be expressed as follows.
\begin{thm}\label{env:innerflucthm}
For a spectral triple with twisted real structure $(\alg{A},\hilb,D,J,\nu)$, the inner-fluctuated Dirac operator $$D_\omega\coloneqq D+\omega+\varepsilon'\nu J\omega J^{-1}\nu$$ for $\omega=\omega^*\in\Omega^1_D(\alg{A})$ arises from $D$ by implementing the bimodule Morita self-equivalence of $(\alg{A},\hilb,D,J,\nu)$ and requiring that the resulting Dirac operator satisfies \eqn{nueC} with respect to $J$ and $\nu$. Thus the data $(\alg{A},\hilb,D_\omega,J,\nu)$ form a spectral triple with twisted real structure.
\end{thm}

\begin{rmk}
As noted in \cite{BCDS16}, Dirac operators are closed with respect to these inner fluctuations, \ie, $$(D_{\omega})_{\omega'}=D_{\omega''}$$ for $D_\omega\coloneqq D+\omega+\varepsilon'\nu J\omega J^{-1}\nu$ with $\omega,\omega''\in\Omega_D^1(\alg{A})$ and $\omega'\in\Omega_{D_{\omega}}^1(\alg{A})$. Explicitly, if $\omega=\sum_i\pi(a_i)\comm{D}{\pi(b_i)}$ and $\omega'=\sum_j\pi(c_j)\comm{D_{\omega}}{\pi(d_j)}$ for $a_i,b_i,c_i,d_i\in\alg{A}$, then one finds (suppressing representations and summations/indices) that $\omega''$ is given by $$\omega''=(a-cda)\comm{D}{b}+(c-cab)\comm{D}{d}+cd\comm{D}{bd}.$$
\end{rmk}

All that remains is to verify that the inner-fluctuated Dirac operator $D_\omega$ satisfies the twisted first-order condition (assuming the original Dirac operator $D$ does as well). 

\begin{prop}
If $(\alg{A},\hilb,D,J,\nu)$ is a spectral triple with twisted real structure, then the fluctuated Dirac operator $D_\omega=D+\omega+\varepsilon'\nu J\omega J^{-1}\nu$ for $\omega\in\Omega^1_D(\alg{A})$ satisfies \eqn{nu1C} with respect to the real structure $J$ and twist operator $\nu$. 
\end{prop}
\begin{proof}
In order for $D_\omega$ to satisfy \eqn{nu1C}, it is sufficient for each of the three summands to individually satisfy \eqn{nu1C}. As $(\alg{A},\hilb,D,J,\nu)$ is a spectral triple with twisted real structure, $D$ satisfies \eqn{nu1C} by assumption. 

Since $\omega\in\Omega^1_D(\alg{A})$, we can write it in the form $\omega=\sum_j\pi(c_j)\comm{D}{\pi(d_j)}$ for $c_j,d_j\in\alg{A}$. In that case, for $a\in\alg{A}$ we have (suppressing summations and representations)
\begin{align*}\comm{\omega}{a}&=\comm{c\comm{D}{d}}{a}\\
&=c\comm{D}{da}-cd\comm{D}{a}-ac\comm{D}{d}.\end{align*}
Each three of these terms satisfies \eqn{nu1C} because $D$ does by assumption. As an example calculation, $D$ satisfying \eqn{nu1C} and $a,d\in\alg{A}$ implies that $\comm{D}{da}J\nu^2b\nu^{-2}J^{-1}=JbJ^{-1}\comm{D}{da}$. Since \eqn{0C} implies that $cJbJ^{-1}=JbJ^{-1}c$, we thus have that $c\comm{D}{da}J\nu^2b\nu^{-2}J^{-1}=JbJ^{-1}c\comm{D}{da}$, \ie, that $c\comm{D}{da}$ satisfies \eqn{nu1C}. The computation for the other terms is carried out in the same way.

For the final term, we can simply make use of Lem.~\ref{env:nuJomegaJnu} to note that $\nu J\omega J^{-1}\nu\in\tilde{\Omega}^1_D(\alg{A}^{\mathrm{op}})$ for $\omega\in\Omega^1_D(\alg{A})$. By the definition of $\tilde{\Omega}^1_D(\alg{A}^{\mathrm{op}})$, any element of $\tilde{\Omega}^1_D(\alg{A}^{\mathrm{op}})$ can be written in the form $\sum_j\pi^*_J(\tilde{\nu}(c_j))\comm{D}{\tilde{\nu}^{-1}(d_j)}^{\pi^*_J}_{\tilde{\nu}^2}$ or, equivalently, $\sum_j\pi_J(\hat{\nu}^{-1}(c_j))\comm{D}{\hat{\nu}(d_j)}^{\pi_J}_{\hat{\nu}^{-2}}$ for $c_j,d_j\in\alg{A}$. We now compute (again suppressing summations and representations) that $$\comm{J\hat{\nu}^{-1}(c)J^{-1}\comm{D}{\hat{\nu}(d)}^{\pi_J}_{\hat{\nu}^{-2}}}{a}=J\hat{\nu}^{-1}(c)J^{-1}\comm{\comm{D}{\hat{\nu}(d)}^{\pi_J}_{\hat{\nu}^{-2}}}{a}$$ by \eqn{0C}. However, from Lem.~\ref{env:Jacobi} we have that $\comm{\comm{D}{\hat{\nu}(d)}^{\pi_J}_{\hat{\nu}^{-2}}}{a}=\comm{\comm{D}{a}}{\hat{\nu}(d)}^{\pi_J}_{\hat{\nu}^{-2}}$ and since $D$ satisfies \eqn{nu1C}, we find $\comm{\comm{D}{a}}{\hat{\nu}(d)}^{\pi_J}_{\hat{\nu}^{-2}}=0$ by Lem.~\ref{env:altnu1C}, and thus \eqn{nu1C} is satisfied by $\nu J\omega J^{-1}\nu$ trivially. 
\end{proof}

The general procedure described in this section relies on defining connections on projective modules. In the setting of Hopf-Galois extensions (`quantum principal bundles'), strong connections similarly induce connections on associated projective modules \cite{DGH01}. It would be interesting to investigate the link between fluctuations of the Dirac operator in the various approaches to Connes' noncommutative geometry (including spectral triples with (twisted) real structure, real twisted spectral triples, real spectral triples without the first-order condition, \etc) and strong connections for Hopf-Galois extensions in scenarios where the formalisms are compatible. As such an investigation is beyond the scope of this paper, however, we leave it for future research.

\section{Gauge transformations of spectral triples with twisted real structure}

\subsection{Gauge transformations as a Morita self-equivalence}

In the context of noncommutative geometry, gauge transformations can be viewed as a special case of Morita self-equivalences, which we will describe briefly as follows. Considering a Morita equivalence between two algebras $\alg{A}$ (acting on $\hilb$) and $\alg{B}$, we have a Morita equivalence bimodule $\EBA$ and the conjugate bimodule $\conjEBA$ (such that $\alg{B}$ acts on $\EBA\otimes_{\alg{A}}\hilb\otimes_{\alg{A}}\conjEBA$ as a Hilbert space). As mentioned in the previous section, when $\EBA$ is finite projective as a right $\alg{A}$-module (which we will assume it to be), we have $\alg{B}\simeq\End{\alg{A}}{\rmod{\alg{E}}{\alg{A}}}$ (and similarly in the left module case). Consider the group of unitary endomorphisms $\alg{U}(\alg{E})\coloneqq\{u\in\End{\alg{A}}{\alg{E}}:uu^*=u^*u=\id_\alg{E}\}$ (with $\alg{U}(\bar{\alg{E}})$ defined likewise). Then for $u\in\alg{U}(\alg{E})\simeq\alg{U}(\bar{\alg{E}})$ we call a \emph{gauge transformation} the action of $u$ on $\EBA$ ($u(e)=ue$) and on $\conjEBA$ ($u(\bar{e})=\widebar{eu^*}$).

In the case of self-equivalence, $\EBA=\conjEBA=\alg{A}$ and $u\in\alg{U}(\alg{A})$. Then gauge transformations on the algebra can be interpreted as inner automorphisms $a\mapsto uau^*$ for all $a\in\alg{A}$. This in turn gives gauge transformations on the Hilbert space $\alg{A}\otimes_{\alg{A}}\hilb\otimes_{\alg{A}}\alg{A}\simeq\hilb$ as $\psi\mapsto u\psi u^*$ for all $\psi\in\hilb$. Here, we already see a divergence from the standard construction due to the modified choice of right action \eqn{psia}, as
\begin{equation}\label{eq:Adu} u\psi u^*\eqqcolon\Ad(u)\psi=\pi(u)(u^*)^\ominus\psi=\pi(u)J\nu\pi(u)\nu^{-1}J^{-1}\psi \end{equation}
for all $\psi\in\hilb$, which contrasts with the usual case where $u\psi u^*=\pi(u)J\pi(u)J^{-1}\psi$. Note that in terms of the notation we have developed, this has the inverse $$\Ad(u)^{-1}\psi=u^*\psi u=\pi(u^*)u^\ominus\psi=\pi(u^*)J\nu\pi(u^*)\nu^{-1}J^{-1}\psi.$$ In addition to the adjoint action above, by $\widetilde{\Ad}(u)$ we denote the complementary `twisted' adjoint action 
\begin{equation}\label{eq:twistedAdu}\widetilde{\Ad}(u)\psi\coloneqq u\psi\tilde{\nu}^2(u^*)=\pi(u)(u^*)^\oplus\psi=\pi(u)J\nu^{-1}\pi(u)\nu J^{-1}\psi\end{equation} for all $\psi\in\hilb$, with inverse
\[ \widetilde{\Ad}(u)^{-1}\psi=\pi(u^*)J\nu^{-1}\pi(u^*)\nu J^{-1}\psi. \]
A peculiarity of \eqn{Adu} (and \eqn{twistedAdu}) is that in general, neither $\Ad(u)$ nor $\widetilde{\Ad}(u)$ is a unitary operator despite being generated by a unitary element, due to the presence of the twist. This means that in the context of twisted real structures, Morita equivalences do not generally give unitary equivalences of spectral triples (which will be established explicitly in Prop. \ref{env:VDW}).

Gauge transformations are also defined on (right $\alg{A}$-module) connections $\nabla\colon\alg{E}\to\alg{E}\otimes_{\alg{A}}\Omega^1(\alg{A})$ as maps
\begin{equation}\label{eq:conngaugetf} \nabla\mapsto\nabla^u\coloneqq u\nabla u^* \end{equation}
for $u\in\alg{U}(\alg{E})$, where $\alg{U}(\alg{E})$ acts on $\EBA\otimes_{\alg{A}}\Omega^1(\alg{A})$ by $u\otimes\id$. Note that the gauge-transformed connection $\nabla^u$ is itself a connection for any $u\in\alg{U}(\alg{E})$ and connection $\nabla$. 

For the case of right module Morita \emph{self}-equivalence, we have $\tensor{\alg{E}}{_{\alg{A}}}=\alg{A}$ and the connection is given by $$\nabla=\delta+\omega$$ for $\omega\in\Omega^1_D(\alg{A})$ self-adjoint. In this context, the connection 1-form $\omega$ is referred to as a \emph{gauge potential}, and a gauge transformation maps the connection to $$\nabla^u=\delta+\omega^u$$ where the gauge transformation is fully encoded in the transformation law for the gauge potential\footnote{For $a\in\alg{A}$, $\nabla^u(a)=u\cdot\delta(u^*a)+u\cdot\omega(u^*a)=\big(\delta+(u\cdot\omega\cdot u^*)+(u\cdot\delta(u^*))\big)(a)\equiv(\delta+\omega^u)(a)$.} 
\begin{equation}\label{eq:gauge1formtf}\omega\mapsto\omega^u=\pi(u)\omega\pi(u^*)+\pi(u)\comm{D}{\pi(u^*)}.\end{equation} 
In the case of right module self-equivalence, the implementation of gauge transformations on the Dirac operator amounts to substituting the gauge transformed connection $\nabla\mapsto\nabla^u$ into the definition of $D_R$. Thus for a gauge transformation by the unitary $u\in\alg{A}$, $D_R=D+\omega$ is mapped to $$D_R^u=D+\omega^u=D+\pi(u)\omega\pi(u^*)+\pi(u)\comm{D}{\pi(u^*)}.$$

As we have already seen in Lem. \ref{env:nuJomegaJnu}, for a given twisted real structure $(J,\nu)$ it is possible to find a `twisted-opposite' 1-form $\omega^\odot$ associated to the 1-form $\omega$ using the map $\omega\mapsto\varepsilon'\nu J\omega J^{-1}\nu$. However, it is not clear that this map is compatible with gauge transformation of gauge potentials \eqn{gauge1formtf}, or in other words, it is not clear if $(\omega^\odot)^u$ coming from the action of $u$ on the connection $\nabla^\odot$ is actually equal to $(\omega^u)^\odot$. In what follows, we show that this is indeed the case.
\begin{prop}\label{env:uodot}
Let $\omega\in\Omega^1_D(\alg{A})$ be a gauge potential. Under the gauge transformation $\omega\mapsto\omega^u$, given in \eqn{gauge1formtf} for $u\in\alg{U}(\alg{A})$, the corresponding twisted 1-form under the map $\omega^u\mapsto(\omega^u)^\odot$ is given by
\[ (\omega^u)^\odot=(u^*)^\oplus\omega^\odot u^\ominus+(u^*)^\oplus\delta^\odot(u). \]
\end{prop}
\begin{proof}
The proof is by a straightforward computation. We will omit representations for brevity. 
\begin{align*}
(\omega^u)^\odot&=\varepsilon'\nu J\omega^uJ^{-1}\nu\\
&=\varepsilon'\nu Ju\omega u^*J^{-1}\nu+\varepsilon'\nu Ju\comm{D}{u^*}J^{-1}\nu\\
&=\varepsilon'J\nu^{-1}(u^*)^*\nu J^{-1}\nu J\omega J^{-1}\nu J\nu u^*\nu^{-1}J^{-1}+\varepsilon'J\nu^{-1}(u^*)^*\nu J^{-1}\nu J\comm{D}{u^*}J^{-1}\nu\\
&=(u^*)^\oplus\omega^\odot u^\ominus+(u^*)^\oplus(\varepsilon'\nu J\comm{D}{u^*}J^{-1}\nu)\\
&=(u^*)^\oplus\omega^\odot u^\ominus+(u^*)^\oplus\delta^\odot(u),
\end{align*}
where in going to the last line we have made use of \eqn{newder} and Lem. \ref{env:nuJomegaJnu} (there taking $a=1,b=u$).
\end{proof}

To compare with the result of Prop. \ref{env:uodot}, we now derive a transformation law from Morita self-equivalence as we did in the (standard) right module case above. For the case of left module self-equivalence, we have that the conjugate module is $\conjEBA=\alg{A}$ with $\alg{B}=\alg{A}$, and the unitary endomorphisms are just the unitary elements of the algebra acting by
\begin{equation}\label{eq:endoalg}u(a)\coloneqq au^*\end{equation}
for $a\in\alg{A}$, $u\in\alg{U}(\alg{A})$.

Given the derivation $\delta^\odot$ as defined in \eqn{newder}, the gauge transformation for the connection $\nabla^\odot=\delta^\odot+\omega^\odot$ on the module $\conjEBA=\alg{A}$, $\alg{B}=\alg{A}$ is given by $\nabla^\odot\mapsto(\nabla^\odot)^u\coloneqq u\nabla^\odot u^*$, analogous to \eqn{conngaugetf}.

\begin{lem}
Let $\nabla^\odot$ be an $\tilde{\Omega}^1_D(\alg{A}^{\mathrm{op}})$-valued connection on $\alg{A}$ as an $\alg{A}$-bimodule with unitary endomorphisms $u$ acting via \eqn{endoalg}. Then, for any $u\in\alg{U}(\alg{A})$ we have
$$( \id \otimes\tilde{\nu}_\alg{E}^2)\circ(\nabla^\odot)^u(a)=\delta^\odot(au)\otimes\tilde{\nu}^2(u^*)+\omega^\odot(au)^\ominus\otimes\tilde{\nu}^2(u^*)$$ for all $a\in\alg{A}$ with $(\nabla^\odot)^u$ the gauge transformation $(\nabla^\odot)^u\coloneqq u\nabla^\odot u^*$ and $\tilde{\nu}_\alg{E}\in\End{\alg{A}}{\alg{A}}$ as given in \eqn{oldmu}.
\end{lem}
\begin{proof}
First, we have $\nabla^\odot=\nabla^\odot_0+\vec{\omega}^\odot$ given by $\nabla^\odot_0(a)=\delta^\odot(a)\otimes1$ and $\vec{\omega}^\odot(a)=(\omega^\odot a^\ominus)\otimes1$ for any $a\in\alg{A}$. We also have $u(a)=au^*$, and thus we find 
\begin{align*}
(\nabla^\odot)^u(a)\coloneqq(u\nabla^\odot u^*)(a)&=(u\nabla^\odot)(au)=u\big(\nabla^\odot_0(au)+\vec{\omega}^\odot(au)\big)\\
&=u\big(\delta^\odot(au)\otimes1+\omega^\odot(au)^\ominus\otimes1\big)\\
&=\delta^\odot(au)\cdot u^*\otimes1+(\omega^\odot(au)^\ominus)\cdot u^*\otimes1\\
&=\delta^\odot(au)\otimes u^*+\omega^\odot(au)^\ominus\otimes u^*.
\end{align*} Therefore, we have
\begin{align*}
( \id \otimes\tilde{\nu}_\alg{E}^2)\circ(\nabla^\odot)^u(a)&=( \id \otimes\tilde{\nu}_\alg{E}^2)(\delta^\odot(au)\otimes u^*+\omega^\odot(au)^\ominus\otimes u^*)\\
&=\delta^\odot(au)\otimes\tilde{\nu}^2(u^*)+\omega^\odot(au)^\ominus\otimes\tilde{\nu}^2(u^*).
\end{align*}
\end{proof}

Just like in the right module case, we implement the gauge transformation for the left module case by replacing $\nabla^\odot$ with $(\nabla^\odot)^u$ in the definition of $D_L$. In the case of self-equivalence, we obtain the following explicit formula:
\begin{prop}\label{env:odotu}
For a gauge transformation with $u\in\alg{U}(\alg{A})$, the operator $D_L=D+\omega^\odot$ is mapped to $D_L^u=D+(\omega^\odot)^u$ where the transformed 1-form is given by $$(\omega^\odot)^u=(u^*)^\oplus\omega^\odot u^\ominus+(u^*)^\oplus\delta^\odot(u).$$
\end{prop}
\begin{proof}
We have from the previous lemma
\begin{align*}
D^u_L(\psi\otimes a)&=(( \id \otimes\tilde{\nu}_\alg{E}^2)\circ(D_L+(\nabla^\odot)^u))(\psi\otimes a)\\
&=D\psi\otimes\tilde{\nu}^2(a)+\delta^\odot(au)\psi\otimes\tilde{\nu}^2(u^*)+\omega^\odot(au)^\ominus\psi\otimes\tilde{\nu}^2(u^*)\\
&=\left(a^\oplus D+(u^*)^\oplus D(au)^\ominus-(u^*)^\oplus(au)^\oplus D\right)\psi\otimes1+(u^*)^\oplus\omega^\odot u^\ominus a^\ominus\psi\otimes1\\
&=(u^*)^\oplus Du^\ominus a^\ominus\psi\otimes1+(u^*)^\oplus\omega^\odot u^\ominus a^\ominus\psi\otimes1\\
&=\left(D+(u^*)^\oplus(Du^\ominus-u^\oplus D)\right)a^\ominus\psi\otimes1+(u^*)^\oplus\omega^\odot u^\ominus a^\ominus\psi\otimes1.
\end{align*} By definition, $Du^\ominus-u^\oplus D=\delta^\odot(u)$. Identifying $\psi\otimes a=a^\ominus\psi\otimes1$ with $\psi\in\hilb\otimes_{\alg{A}}\alg{A}\simeq\hilb$ therefore gives the proposition.
\end{proof}
Comparing the results of Prop. \ref{env:uodot} and Prop. \ref{env:odotu}, we find that indeed, $(\omega^u)^\odot=(\omega^\odot)^u$.

\subsection{Gauge transformations for a spectral triple}

The results of the previous subsection, along with Thm.~\ref{env:innerflucthm}, give the following result:
\begin{thm}\label{env:DwuDuw}
Let $(\alg{A},\hilb,D_\omega,J,\nu)$ be a spectral triple with twisted real structure obtained by bimodule Morita self-equivalence from the spectral triple with twisted real structure $(\alg{A},\hilb,D,J,\nu)$, where $D_\omega$ is the Dirac operator fluctuated from $D$ by the 1-form $\omega\in\Omega^1_D(\alg{A})$. 

Then the law for the gauge transformation of the Dirac operator $D_\omega$ by $u\in\alg{U}(\alg{A})$ is given by $$D_\omega\mapsto D_\omega^u=D+\omega^u+\varepsilon'\nu J\omega^uJ^{-1}\nu\equiv D_{\omega^u},$$ where $$\omega\mapsto\omega^u=\pi(u)\omega\pi( u^*)+\pi(u)\comm{D}{\pi(u^*)}$$ is the gauge transformation of a gauge potential.
\end{thm}
The above expression is found by mapping $\omega\mapsto\omega^u$ on the operator $D_\omega$ after $\omega_L$ and $\omega_R$ have been identified. For the sake of consistency, one should check that the same result applies for gauge transforming both left and right gauge potentials separately, and indeed this does prove to be the case, though we will not give the proof here.

Similar to the standard case, the gauge transformation of the Dirac operator can be implemented using the operator which implements gauge transformations of the Hilbert space \eqn{Adu}. The difference is that we require not just the operator for the adjoint action, but also the corresponding operator for the twisted adjoint action \eqn{twistedAdu}.
\begin{lem}\label{env:puregauge}
Let $(\alg{A},\hilb,D,J,\nu)$ be a spectral triple with twisted real structure. For any $u\in\alg{U}(\alg{A})$ it holds that $$\widetilde{\Ad}(u)D\Ad(u)^{-1}=D+\pi(u)\comm{D}{\pi(u^*)}+\varepsilon'\nu J\pi(u)\comm{D}{\pi(u^*)}J^{-1}\nu.$$
\end{lem}
\begin{proof}
The proof is given by a straightforward computation (we suppress representations for brevity):
\begin{align*}
\widetilde{\Ad}(u)D\Ad(u)^{-1}&=uJ\nu^{-1}u\nu J^{-1}DJ\nu u^*\nu^{-1}J^{-1}u^*\\
&=\varepsilon'u\nu JuDu^*J^{-1}\nu u^*\\
&=\varepsilon'u\nu J(D+u\comm{D}{u^*})J^{-1}\nu u^*\\
&=uDu^*+\varepsilon'u\nu Ju\comm{D}{u^*}J^{-1}\nu u^*\\
&=D+u\comm{D}{u^*}+\varepsilon'u\nu Ju\comm{D}{u^*}J^{-1}\nu u^*.
\end{align*}
Now all that remains is to massage the final term on the right hand side:
\begin{align*}
\varepsilon'u\nu Ju\comm{D}{u^*}J^{-1}\nu u^*&=\varepsilon'\nu J\nu J^{-1}u\nu Ju\comm{D}{u^*}J^{-1}\nu u^*J\nu J^{-1}\nu\\
&=\varepsilon'\nu J(J\nu^{-1}u\nu J^{-1})u\comm{D}{u^*}(J\nu u^*\nu^{-1}J^{-1})J^{-1}\nu\\
&=\varepsilon'\nu Ju\comm{D}{u^*}(J\nu u\nu^{-1}J^{-1})(J\nu u^*\nu^{-1}J^{-1})J^{-1}\nu\\
&=\varepsilon'\nu Ju\comm{D}{u^*}J^{-1}\nu,
\end{align*}thus giving the claimed result. Note that on the third line, we have used \eqn{0C} and \eqn{nu1C}.
\end{proof}
\begin{note}
We refer to Dirac operators obtained by gauge transformations of a non-fluctuated Dirac operator as being \emph{pure gauge}.
\end{note}

\begin{prop}\label{env:VDW}
Let $(\alg{A},\hilb,D,J,\nu)$ be a spectral triple with twisted real structure, and consider a fluctuated Dirac operator $D_\omega=D+\omega+\varepsilon'\nu J\omega J^{-1}\nu$.  Then for any $u\in\alg{U}(\alg{A})$, one has $$\widetilde{\Ad}(u)D_\omega\Ad(u)^{-1}=D+\omega^u+\varepsilon'\nu J\omega^uJ^{-1}\nu,$$ with the gauge transformed $\omega^u$ given as above.
\end{prop}
\begin{proof}
Without loss of generality, we can take $\omega=a\comm{D}{b}$, where we are again suppressing representations for brevity. In that case, we have 
\begin{align*}
\widetilde{\Ad}(u)\omega\Ad(u)^{-1}&=uJ\nu^{-1}u\nu J^{-1}a\comm{D}{b}J\nu u^*\nu^{-1}J^{-1}u^*\\
&=ua\comm{D}{b}(J\nu u\nu^{-1}J^{-1})(J\nu u^*\nu^{-1}J^{-1})u^*\\
&=ua\comm{D}{b}u^*=u\omega u^*,
\end{align*} and a slightly more involved (but not qualitatively different) calculation gives
\begin{align*}
\widetilde{\Ad}(u)(\varepsilon'\nu J\omega J^{-1}\nu)\Ad(u)^{-1}&=\varepsilon'uJ\nu^{-1}u\nu J^{-1}\nu Ja\comm{D}{b} J^{-1}\nu J\nu u^*\nu^{-1}J^{-1}u^*\\
&=\varepsilon'u\nu Jua\comm{D}{b}u^*J^{-1}\nu u^*\\
&=\varepsilon'\nu J\nu J^{-1}u\nu Jua\comm{D}{b}u^*J^{-1}\nu u^*J\nu J^{-1}\nu\\
&=\varepsilon'\nu J(J\nu^{-1}u\nu J^{-1})ua\comm{D}{b}u^*(J\nu u^*\nu^{-1}J^{-1})J^{-1}\nu\\
&=\varepsilon'\nu Jua\comm{D}{b}u^*(J\nu u\nu^{-1}J^{-1})(J\nu u^*\nu^{-1}J^{-1})J^{-1}\nu\\
&=\varepsilon'\nu Jua\comm{D}{b}u^*J^{-1}\nu=\varepsilon'\nu Ju\omega u^*J^{-1}\nu.
\end{align*} Collecting these results and combining with Lem.~\ref{env:puregauge}, one finds the claimed result.
\end{proof}

\begin{prop}\label{env:VW}
Let $(\alg{A},\hilb,D,J,\gamma,\nu)$ be an even spectral triple with twisted real structure.  Then for any $u\in\alg{U}(\alg{A})$, the gauge transformation of the spectral triple by $u$ is characterised by the following operator actions:
\begin{align}
\tilde{V}\pi(a)\tilde{V}^{-1}=V\pi(a)V^{-1}&=\pi(uau^*),\label{eq:atouau}\\
V\psi&=\psi^u\\
\tilde{V}DV^{-1}&=D^u,\label{eq:Durule}\\
\tilde{V}\gamma\tilde{V}^{-1}=V\gamma V^{-1}&=\gamma,\label{eq:gammarels}\\
\tilde{V}\nu J\tilde{V}^{-1}&=\nu J,\label{eq:intertwinenuJ}\\
VJ\nu V^{-1}&=J\nu,\label{eq:intertwineJnu}
\end{align} for all $a\in\alg{A}$ and $\tilde{V}\coloneqq\widetilde{\Ad}(u)=\pi(u)J\nu^{-1}\pi(u)\nu J^{-1}$ and $V\coloneqq\Ad(u)=\pi(u)J\nu\pi(u)\nu^{-1}J^{-1}$.
\end{prop}
\begin{proof}
We have that $V\psi=\psi^u$ from the definition of gauge transformations. This is also true of $\pi(a)^u=\pi(uau^*)$, which one checks is given by conjugation by both $V$ and $\tilde{V}$ by computation. That $\tilde{V}DV^{-1}=D^u$ is given by Prop.~\ref{env:VDW}. All of the other relations are obtained by straightforward computations. 
\end{proof}
The invariance of $\gamma$ under both $V$ and $\tilde{V}$ nicely dovetails with the gauge transformations of $\alg{A}$ and $D$, which follows from the earlier constructions (specifically, $\gamma$ must commute with $\alg{A}$ and anticommute with $D$). However, it is curious that the Dirac operator does not generally respect gauge transformations on the Hilbert space, \eg, $D^u\psi^u=\tilde{V}DV^{-1}V\psi=\tilde{V}D\psi\neq(D\psi)^u$ when $\tilde{V}\neq V$. The meaning of this is not entirely clear, although given that the choice of right $\alg{A}$-action \eqn{psia} was somewhat arbitrary (we could equally have chosen $\psi a\coloneqq a^\oplus\psi$), it may be a hint that our construction of Morita equivalence is in some way half-complete, and that we should have two gauge transformations, one for each twist (or perhaps even for an infinite number of twists coming from $\tilde{\nu}^n$, $n\in\ZZ$). However, as the construction presented in this paper is sufficient for our purposes, leave this possibility for further investigation.

Another interesting point is that the real structure $J$ (and hence the twist operator $\nu$) does not individually obey any invariance with respect to $V$, $\tilde{V}$, or their inverses, as we might hope. However, the pair of equations \eqn{intertwinenuJ} and \eqn{intertwineJnu} do make sense as the weakest way of ensuring \eqn{Durule} is compatible with \eqn{nueC}. The most natural inference from the pair \eqn{intertwinenuJ} and \eqn{intertwineJnu} is to accept that there is a covariance instead of an invariance, in which case we should define $J^u=VJ\tilde{V}^{-1}$ and $\nu^u=\tilde{V}\nu V^{-1}$. But in that case $J^u$ can only possibly be a real structure if $\nu$ is both self-adjoint and unitary (up to sign). This should be borne in mind for the following sections.

It is of course also important to note that, for $U\coloneqq\pi(u)J\pi(u)J^{-1}$, when $\nu=1$ (the trivially-twisted case) or $\nu\pi(u)\nu^{-1}=\pi(u)$ for all $u\in\alg{U}(\alg{A})$, we have that $\tilde{V}=V=U$ and $V^{-1}=\tilde{V}^{-1}=U^*$ and all of the above relations reduce to the familiar ones for unitary equivalences of spectral triples. 
\begin{rmk}
A curious observation is that when $u$ is permitted to remain invariant under $\hat{\nu}$, or equivalently, we request that $\comm{\nu}{\pi(u)}=0$, we see that $\tilde{V}$ reduces to $\pi(u)J\pi(u)J^{-1}$ and $V^{-1}$ to $\pi(u^*)J\pi(u^*)J^{-1}$, which are simply the familiar unitary operators which implement gauge transformations in the trivially-twisted case. This is interesting because in principle the remainder of the twisted real structure remains intact, unlike in the case of Prop.~\ref{env:ruinerprop}, for example.\footnote{In that case, the spectral triple with twisted real structure $(J,\nu=\pm\nu^*=\pm\nu^{-1})$ is equivalent to the trivially-twisted spectral triple with real structure $\mathcal{J}=\nu J$. Unitary equivalence would then be restored automatically with $V=\mathcal{U}$ for $\mathcal{U}=\pi(u)\mathcal{J}\pi(u)\mathcal{J}^{-1}$.} That said, it is likely not practical to impose such invariance, as it would almost surely be far too restrictive on either the
  available twists or the usable unitary elements.
\end{rmk}

\subsection{Self-adjointness of the Dirac operator}

In the trivially-twisted case, for $U\coloneqq\pi(u)J\pi(u)J^{-1}$, a gauge transformation preserves the self-adjointness of the Dirac operator automatically. The transformed operator $D_{\omega^u}=UD_\omega U^*$ is self-adjoint if and only if $D_\omega$ is, since $U$ is unitary, and so a gauge transformation yields a spectral triple which is unitarily equivalent to the former. This is no longer necessarily true when the real structure is twisted. We now investigate the cases in which it is.

\begin{lem}
If $(\alg{A},\hilb,D,J,\nu)$ is a spectral triple with twisted real structure, and $D^u$ is the Dirac operator obtained from $D$ by a gauge transformation by the unitary element $u\in\alg{U}(\alg{A})$ then $$\nu=\pm\nu^*$$ is a sufficient condition for $D^u=(D^u)^*$. 
\end{lem}
\begin{proof}
From Lem.~\ref{env:puregauge} we know that $D^u=\tilde{V}DV^{-1}$ for $\tilde{V}=\pi(u)J\nu^{-1}\pi(u)\nu J^{-1}$ and $V^{-1}=\pi(u^*)J\nu\pi(u^*)\nu^{-1}J^{-1}$. Since $D$ is self-adjoint by the definition of a Dirac operator, we therefore have that $D^u=(D^u)^*$ is given by $\tilde{V}DV^{-1}=(V^{-1})^*D\tilde{V}^*$ or, expressed more fully, $$uJ\nu^{-1}u\nu J^{-1}DJ\nu u^*\nu^{-1}J^{-1}u^*=uJ(\nu^{-1})^*u\nu^*J^{-1}DJ\nu^*u^*(\nu^{-1})^*J^{-1}u^*.$$ It is clear from simple substitution that $\nu=\pm\nu^*$ satisfies this equation.
\end{proof}
Note that when $\nu=\pm\nu^*$, one has that $\tilde{V}^*=V^{-1}$ and $V^*=\tilde{V}^{-1}$. 
\begin{rmk}
The requirement that the twist operator be self-adjoint (up to sign) is well-motivated by comparison to the literature on (real) twisted spectral triples, where it is equivalent to the common requirement (going back to \cite{CM08}) that $\rho(a^*)=(\rho^{-1}(a))^*$, where $\rho$ is the twisting algebra automorphism\footnote{Recall that $\tilde{\nu}(a^*)=(\hat{\nu}^{-1}(a))^*$ and $\tilde{\nu}\equiv\hat{\nu}$ exactly when $\nu^*=\nu$ (up to sign).} ($\rho$ can be seen as very roughly analogous to $\hat{\nu}^2$ or $\tilde{\nu}^2$ in the twisted real structure context, though of course the frameworks differ).\end{rmk}

\begin{prop}
Let $(\alg{A},\hilb,D,J,\nu)$ be a spectral triple with twisted real structure, with $D_\omega$ the Dirac operator obtained by fluctuating $D$ by a 1-form $\omega\in\Omega^1_D(\alg{A})$ and $D_\omega^u$ the Dirac operator obtained from $D_\omega$ by a gauge transformation by the unitary element $u\in\alg{U}(\alg{A})$. Then, for $\nu=\pm\nu^*$, $$\omega=\omega^*$$ is a sufficient condition for $D_\omega^u=(D_\omega^u)^*$. 
\end{prop}
\begin{proof}
We begin by considering the Dirac operator $D_\omega=D+\omega+\varepsilon'\nu J\omega J^{-1}\nu$. As $D$ is self-adjoint by assumption, the self-adjointness of $D_\omega$ is guaranteed by the condition that $$\omega+\varepsilon'\nu J\omega J^{-1}\nu=\omega^*+\varepsilon'\nu^*J\omega^*J^{-1}\nu^*.$$ Taking $\nu=\pm\nu^*$, this can be rewritten $$(\omega-\omega^*)+\varepsilon'\nu J(\omega-\omega^*)J^{-1}\nu=0,$$ which is clearly satisfied when $\omega=\omega^*$.

By Thm.~\ref{env:DwuDuw}, $D^u_\omega=D_{\omega^u}$ where $\omega^u=\pi(u)\omega\pi(u^*)+\pi(u)\comm{D}{\pi(u^*)}$. The same reasoning therefore applies with $\omega^u$ replacing $\omega$, and so $D^u_\omega$ will be self-adjoint when $\nu=\pm\nu^*$ and $\omega^u=(\omega^u)^*$. However, $\omega=\omega^*$ immediately implies $\omega^u=(\omega^u)^*$.
\end{proof}

Strictly speaking, it is not necessary to take $\nu=\pm\nu^*$ for the twisted real structure formulation to be fully self-consistent. For example, in the case of the trivial Dirac operator $D=0$, which also has trivial fluctuations, one is free to take a non-self-adjoint twist operator. However, as we are interested in the general case, assuming self-adjointness (up to sign) proves to be the only really practical option to guarantee everything will work.

\section{The spectral action for spectral triples with twisted real structures}

In light of the results of the previous section, from here on we will take ${\nu=\alpha_1\nu^*}$ for $\alpha_1\in\{-1,+1\}$. When considering twisted real structures, we have no reason to expect that the standard bilinear form which gives the fermionic action $\mathfrak{A}_D(\psi,\phi)\coloneqq\mip{J\psi}{D\phi}$, $\psi,\phi\in\Dom(D)$, should still hold unchanged. Indeed, in the twisted real structure setting, this bilinear form fails to be suitably (anti)symmetric; we find 
\begin{align}
\mathfrak{A}_D(\psi,\phi)&=\mip{J\psi}{D\phi}\nonumber\\
&=\varepsilon\mip{J\psi}{J^2D\phi}\nonumber\\
&=\varepsilon\mip{JD\phi}{\psi}\nonumber\\
&=\varepsilon\varepsilon'\mip{\nu^{-1}DJ\nu\phi}{\psi}\nonumber\\
&=\varepsilon\varepsilon'\mip{J\phi}{{\nu^{-1}}^*D^*{\nu^{-1}}^*\psi}\label{eq:rubbish}
\end{align}
 which is not equal to $\mathfrak{A}_D(\phi,\psi)$ (up to sign) unless $D=\pm{\nu^{-1}}D^*{\nu^{-1}}$, which is not true in general.

Therefore our first task is to see if we can construct an alternative bilinear form which is gauge-covariant, and indeed, this can be done:
\begin{lem}\label{env:bilinearform}
Let $\nu$ be a linear operator. Then the bilinear form
\begin{equation}\label{eq:Anubar}\tilde{\mathfrak{A}}_D(\psi,\phi)\coloneqq\mip{J\nu\psi}{D\phi}_{\hilb},\end{equation}
where $\mip{\place}{\place}_{\hilb}$ is the inner product on the Hilbert space $\hilb$, is gauge-covariant and is correctly (anti)symmetric for the appropriate KO-dimension.\footnote{The (anti)symmetry of the form is important when interpreting the bracket in the physical context of the Grassmann nature of fermion fields.}
\end{lem}
\begin{proof}
That the bilinear form be gauge-covariant is equivalent to requiring that it satisfies $$\tilde{\mathfrak{A}}_D(\psi,\phi)=\tilde{\mathfrak{A}}_{D^u}\big(\mathrm{Ad}(u)(\psi),\mathrm{Ad}(u)(\phi)\big).$$ Using the notation $V\coloneqq\Ad(u)=\pi(u)J\nu\pi(u)\nu^{-1}J^{-1}$ and $\tilde{V}=\pi(u)J\nu^{-1}\pi(u)\nu J^{-1}$, one finds this is equivalent to the requirement that 
\begin{align*}\mip{J\nu\psi}{D\phi}&=\mip{J\nu V\psi}{D^uV\phi}\\
&=\mip{J\nu V\psi}{\tilde{V}DV^{-1}V\phi}\\
&=\mip{JV\psi}{\tilde{V}D\phi},\end{align*} 
which in turn is equivalent to requiring that $$\tilde{V}^*J\nu V=J\nu,$$ which one can easily compute to be true.
The (anti)symmetry of the bilinear form can also be verified by computation. Assuming that $D$ is self-adjoint, we have 
\begin{align*}
\tilde{\mathfrak{A}}_D(\psi,\phi)&=\mip{J\nu\psi}{D\phi}\\
&=\varepsilon\mip{J\nu\psi}{J^2D\phi}\\
&=\varepsilon\mip{JD\phi}{\nu\psi}\\
&=\alpha_1\varepsilon\mip{\nu JD\phi}{\psi}\\
&=\alpha_1\varepsilon\varepsilon'\mip{DJ\nu\phi}{\psi}\\
&=\alpha_1\varepsilon\varepsilon'\mip{J\nu\phi}{D\psi}\\
&=\alpha_1\varepsilon\varepsilon'\tilde{\mathfrak{A}}_D(\phi,\psi),
\end{align*}
which can always be made to have the same sign as in the untwisted case (for example, by setting $\alpha_1=+1$).
\end{proof}

\begin{rmk}
Equation \eqn{Anubar} can be written in the form $\tilde{\mathfrak{A}}_D(\psi,\phi)=\alpha_1\mip{J\psi}{\nu^{-1}D\phi}$, which allows one to draw an analogy to the work presented in \cite{DFLM18}, wherein the authors define the sesquilinear form $\mip{\psi}{\phi}_\rho\coloneqq\mip{\psi}{R\phi}$ for a linear operator $R=R^*=R^{-1}$ implementing the algebra automorphism $\rho(\place)\coloneqq R(\place)R^*$. However, for our purposes it is preferable to keep the real structure and twist together, so we choose not to use a similarly modified bracket.
\end{rmk}

The main consequence of Lem. \ref{env:bilinearform} is the following:
\begin{prop}
The appropriate form for the fermionic action functional for spectral triples with twisted real structure is \begin{equation}\label{eq:SFtwist}S_F[D,\psi]\coloneqq\mip{J\nu\tilde{\psi}}{D\tilde{\psi}},\end{equation} where $\tilde{\psi}$ is the Grassmann variable corresponding to $\psi\in\{\varphi\in\hilb:\gamma\varphi=\varphi\}$, which is well-defined and gauge-covariant when $D=D^*$ and $\nu=\alpha_1\nu^*$ and is antisymmetric in KO-dimension $2\mod8$ when $\alpha_1=+1$.
\end{prop}

The case for the bosonic action $$S_B[D]=\mathrm{Tr}\left(f\left(\frac{D}{\Lambda}\right)\right)$$ for $\Lambda\in\mathbb{R}$ is less neat. The action should be invariant with respect to gauge transformations of the Dirac operator, but it is not difficult to see that $D^u$ need not have the same spectrum as $D$ in general,\footnote{If $D\psi=\lambda\psi$, we have $D^u\psi^u=\tilde{V}DV^{-1}V\psi=\lambda\tilde{V}\psi\neq\lambda\psi^u$ (see Prop. \ref{env:VW}).} which is a problem. The simplest fix for this is to require $\nu=\pm\nu^{-1}$, which ensures $D$ and $D^u$ have the same spectrum (as then $\tilde{V}=V$), but one must be careful not to run afoul of Prop. \ref{env:ruinerprop}.

\section{The Standard Model and beyond} 

\subsection{First attempt}

Having gone to the effort of constructing a consistent formulation of gauge transformations using twisted real structures, the natural thing is to try to put it to work. As we found in the previous sections, we require $\nu=\pm\nu^*$ to ensure that fluctuations of the Dirac operator are self-adjoint, and doing so then means we require $\nu=\pm\nu^*$ to ensure the bosonic action is well defined. Unfortunately, the requirement that $\nu=\alpha_1\nu^*=\alpha_2\nu^{-1}$ for $\alpha_1,\alpha_2\in\{-1,+1\}$ is very restrictive. Certainly, at the very least, it forces the twist to be mild, such that \eqn{nu1C} simplifies to \eqn{1C}. Even so, it is worth exploring the realm of applicability.

The first place to look for something new is the Standard Model, or at least, the finite part of its spectral triple. We would like to keep the spectral data $(\ASM,\CC^{96},D_{\mathrm{SM}},J_F,\gamma_F)$ unchanged, as they all carry quite neat physical interpretations, but it is not difficult to see that this does not leave much room for adding a twist that does anything. For one thing, consider how one should simultaneously satisfy $\nu J_FD_{\mathrm{SM}}=\DSM J_F\nu$ and $J_F\DSM=\DSM J_F$ for a nontrivial $\nu$.

The next place to look, as suggested in the Introduction, is the Pati-Salam model of \cite{CCvS13b}, which takes for its spectral data 
\begin{equation}\label{eq:PS}(\ALR,\CC^{96},D_{\mathrm{SM}},J_F,\gamma_F).\end{equation}
It is well known that $\ALR$ does not respect the first-order condition with respect to $\DSM$ and $J_F$, so there might be room for a twisted first-order condition to hold instead. But first we will briefly summarise the approach of \cite{CCvS13b}. 

\subsection{Discarding the first-order condition}

The solution to the problem that the first-order condition is not satisfied which is offered by \cite{CCvS13b} is to discard the first-order condition altogether. This can be done, as explained in \cite{CCvS13a}, by changing the way the Dirac operator fluctuates, so that if one wants to fluctuate a Dirac operator $D$ by a one-form $\omega=\sum_i\pi(a_i)\comm{D}{\pi(b_i)}$, $a_i,b_i\in\alg{A}$, rather than use \eqn{Dalpha}, one instead takes
\begin{align}\label{eq:quadfluc} D_\omega&=\sum_i\pi(a_i)\comm{D}{\pi(b_i)}+\sum_iJ\pi(a_i)J^{-1}\comm{D}{J\pi(b_i)J^{-1}}\nonumber\\
&\qquad+\sum_{i,j}J\pi(a_i)J^{-1}\pi(a_j)\comm{\comm{D}{\pi(b_j)}}{J\pi(b_i)J^{-1}}, \end{align}
the substantial difference being the inclusion of the final non-linear `quadratic' correction term, which of course goes to 0 when one assumes \eqn{1C} holds. 

This quadratic term, which we will denote $\omega_{(2)}$, then gauge transforms under the rule 
\[ \omega_{(2)}^u=J\pi(u)J^{-1}\omega_{(2)}J\pi(u^*)J^{-1}+J\pi(u)J^{-1}\comm{\pi(u)\comm{D}{\pi(u^*)}}{J\pi(u^*)J^{-1}}. \]
Though \cite{CCvS13b} demonstrates that the presence of the quadratic term does not lead to anything radical in terms of the particle content of the theory, that is not to say that we should not tread carefully; the sacrifice of the first-order condition means the Dirac operator can no longer be considered a (noncommutative) first-order differential operator. This in particular may be considered too high a price to pay to push the bounds of applicability of the noncommutative approach to gauge theory, and so we investigate as an alternative if the above described \emph{weakening} of the first-order condition provided by a twisted real structure might serve to obtain some alternative noncommutative description of the Pati-Salam model. 

\subsection{Spectral triples with multitwisted real structure}

The trouble with this approach is that, as we have seen, for the bosonic spectral action to make sense, we need $\nu=\pm\nu^{-1}$ which reduces the twisted first-order condition to the ordinary first-order condition, which we already know does not hold for \eqn{PS}. A loophole is provided by the proposed multitwisted real spectral triples of \cite{DS21}. This formalism is an extension of the twisted real structures introduced in \S2. We will not explain it in depth here, but to summarise, one decomposes the Dirac operator such that 
\[ D=\sum_\ell D_\ell,\quad\ell\in\{1,2,\ldots,N\},\]
and to each component $D_\ell$ associates a twist operator $\nu_\ell$. Therefore, practically speaking, in all of the definitions in \S2 one replaces $\nu$ with $\nu_\ell$ and $D$ with $D_\ell$, additionally replacing \eqn{0C} with the multitwisted zeroth-order condition
\begin{equation}\label{eq:nul0C}
\comm{\pi(a)}{J\nu_\ell\pi(b)\nu_\ell^{-1}J^{-1}}=0=\comm{\pi(a)}{J\nu_\ell^{-1}\pi(b)\nu_\ell J^{-1}}
\end{equation}
for all $\ell$, and replacing \eqn{nu1C} with the multitwisted first-order condition
\begin{equation}\label{eq:nul1C}
\comm{D_\ell}{\pi(a)}J\nu_\ell\pi(b)\nu_\ell^{-1}J^{-1}=J\nu_\ell^{-1}\pi(b)\nu_\ell J^{-1}\comm{D_\ell}{\pi(a)}
\end{equation}
for all $\ell$. These last two changes are not trivial (even when $N=1$) because we also now no longer require that $\nu_\ell\pi(\alg{A})\nu_\ell^{-1}\simeq\alg{A}$, and instead only require that conjugation by $\nu_\ell$ is an automorphism of $\bdd$. 

It is clear that one cannot start from the multitwisted perspective and proceed via the path we have carved out above. For one thing, the formulation of Morita equivalences provided in \S3 clearly will not carry over directly, as this would require some abstract decomposition of connections which would be hard to account for, amongst other difficulties. However, for our purposes, Morita equivalence is only necessary for developing a framework of gauge transformations for spectral triples as given in \S\S4--5, which can be expressed wholly in terms of operators. Once these definitions have been laid out, they can be extended in an entirely analogous fashion to that of extending spectral triples with twisted real structures to spectral triples with multitwisted real structures described above, \ie, everywhere replacing $D\mapsto D_\ell$ and $\nu\mapsto\nu_\ell$.

While it is reasonably straightforward to write down new gauge transformations for the multitwist, it should be remarked upon that this process is not always trivial. For example, suppose 
\[ \Omega^1_D(\alg{A})\ni\omega=\pi(a)\comm{D}{\pi(b)}=\pi(a)\comm\bigg{\sum_\ell D_\ell}{\pi(b)}=\sum_\ell\pi(a)\comm{D_\ell}{\pi(b)}. \]
The only way to make sense of an equivalent version of the map $\omega\mapsto\omega^\odot$ is by defining
\[ \omega^\odot\coloneqq\sum_\ell\nu_\ell J\pi(a)\comm{D_\ell}{\pi(b)}J^{-1}\nu_\ell, \]
but unlike $\omega$, this $\omega^\odot$ cannot be derived from any object built from the complete Dirac operator $D$; for example, the fluctuation $D\mapsto\sum_\ell(D_\ell+a\comm{D_\ell}{b}\pm\sum_{k}\nu_k Ja\comm{D_\ell}{b} J^{-1}\nu_k)$ would not satisfy $\nu_\ell JD_\ell=\pm D_\ell J\nu_\ell$ except when $k=\ell$.

Thus there are two ways in which the spectral Pati-Salam model might be (re)constructed using multitwisted real structures: one is that even for a single twist ($N=1$), we have a broader selection of twists to work with than earlier assumed, and of course the other is the possibility to try to use multiple twists ($N>1$). 

\begin{notn}
Because we no longer necessarily have that $\nu\pi(\alg{A})\nu^{-1}\simeq\pi(\alg{A})$, from here on we will use a slightly different notation for twisted commutators to that described in \S2. We will now take ${\comm{A}{B}_\rho\coloneqq AB-\rho(B)A}$, where here $\rho$ is a map acting directly on operators. As a further point of notation, we will make use the shorthand $\Ad(\nu)\eqqcolon\bar{\nu}$. 
\end{notn}

\subsection{The Pati-Salam case (shown with a toy model)}

The matrices involved in the (finite part of the) spectral Pati-Salam case are quite large and unwieldy, and thus difficult to express written out in full. However, many of the issues that come up when working with them also arise in the simpler toy model given in \cite{CCvS13a}, and so for demonstrative purposes we will largely present that case instead. We take for the algebra $\Atoy=\CC_L\oplus\CC_R\oplus M_2(\CC)$, represented as matrices acting on the Hilbert space $\CC^8$ by
\[ \pi\colon(\lambda_L,\lambda_R,M)\mapsto(\diag(\lambda_L,\lambda_R)\otimes1_2)\oplus(1_2\otimes M). \]
The real structure is then given by $J=\left(\begin{smallmatrix}0&1_4\\1_4&0\end{smallmatrix}\right)\circ\cc$, where $\cc$ denotes complex conjugation. Next we consider the Dirac operator $D$, which is given by \[ D=\begin{pmatrix}S&T^*\\T&\bar{S}\end{pmatrix} \] where $S=\left(\begin{smallmatrix}0&k_x\\ \bar{k}_x&0\end{smallmatrix}\right)\otimes1_2$ and $T=\diag(k_y,0,0,0)$. With respect to this Dirac operator, the ordinary first-order condition \eqn{1C} is only satisfied for the `symmetry-broken' subalgebra $\CC_L\oplus\CC_R\oplus\CC_0\subset\Atoy$. 
First we will briefly investigate if it is possible to satisfy the multitwisted first-order condition \eqn{nul1C} for the unbroken algebra $\Atoy$ instead. 

We start by considering matrices which, for the sake of notational ease, we call $B_\ell^+\coloneqq J\nu_\ell\pi(b)\nu_\ell^{-1}J^{-1}$ and $B_\ell^-\coloneqq J\nu_\ell^{-1}\pi(b)\nu_\ell J^{-1}$ for $b\in\Atoy$, all of which must be elements of $\pi(\Atoy)'$ in order to satisfy \eqn{nul0C}. As such, $B_\ell^\pm$ must take the form $m\oplus n\oplus\diag(\mu_1,\mu_1,\mu_2,\mu_2)$ for $m,n\in M_2(\CC)$, $\mu_1,\mu_2\in\CC$. The question then is whether we can identify some operator(s) $\nu_\ell$ which would allow us to obtain $B_\ell^\pm$ from a given $b\in\Atoy$. 
Assuming we have $N=1$ twists, it is not difficult to compute the one-forms $\comm{D}{\pi(a)}$, $a\in\Atoy$, and it is with respect to such one-forms we can try to impose the usual twisted first-order condition \eqn{nu1C}, \ie, we demand that $\comm{D}{\pi(a)}B^+=B^-\comm{D}{\pi(a)}$. Doing so further restricts $B^\pm$ to be of the form
\begin{subequations}\begin{align}
B^+&=\begin{pmatrix}m_{11}^+&0\\m_{21}^+&m_{22}^+\end{pmatrix}\oplus
\begin{pmatrix}\mu_1^+&n^+_{12}\\0&n^+_{22}\end{pmatrix}\oplus\diag(\mu_1^+,\mu_1^+,\mu_2^+,\mu_2^+),\label{eq:Bpl}\\
B^-&=\begin{pmatrix}\mu_1^+&n_{12}^+\\0&n_{22}^+\end{pmatrix}\oplus
\begin{pmatrix}m^+_{11}&0\\m_{21}^+&m^+_{22}\end{pmatrix}\oplus\diag(m_{11}^+,m_{11}^+,\mu_2^-,\mu_2^-),\label{eq:Bmn}
\end{align}\end{subequations}
where we have tried to express everything in terms of elements of $B^+$. 

Even relaxing the requirement that $\nu^2=1$ (in which case $B^-=B^+$), reading off \eqn{Bpl} and \eqn{Bmn} there is a strong suggestion that $\nu^2$ be given by a pair of blockwise flips (on the first and last pairs of $2\times2$ blocks respectively), but in that case one must still make the identifications $\mu^+_2=m_{11}^+$ and $\mu_2^-=\mu_1^+$, which necessarily breaks the algebra $\Atoy$ down to a subalgebra. Not imposing those identifications, it is not at all clear what (or if any) $\nu$ can be found to relate $B^-$ to $B^+$. When $\nu^2=1$ however, it is immediately clear that $\Atoy$ breaks to $\CC^3$.

The ultimate reason for this breaking of the algebra is $T$ being nonzero; thus, if multitwists are to be applicable, it would make sense to decompose $D$ into $D_1=S\oplus\bar{S}$ and either $D_2=\begin{pmatrix}0&T^*\\0&0\end{pmatrix}$ and $D_3=\begin{pmatrix}0&0\\T&0\end{pmatrix}$, or $D_{2'}\coloneqq D_2+D_3$. Using the same method as before we find that, as expected, the $\ell=1$ case is well behaved, but the same cannot be said for the other components.

\begin{prop}\label{env:2twistfail}
The requirement that
$(\Atoy,\CC^8,D_1+D_{2'},J,\{\nu_1, \nu_{2'}\})$,
where $D_1=\left(\begin{smallmatrix}S&0\\0&\bar{S}\end{smallmatrix}\right)$ and $D_{2'}=\left(\begin{smallmatrix}0&T^*\\T&0\end{smallmatrix}\right)$, with all other data defined as above and twists satisfying $\nu_1^2=\nu_{2'}^2=1$,
be a multitwisted spectral triple breaks the algebra $\Atoy$ to $\CC^3$.
\end{prop} 
\begin{proof}
Since we are taking $\nu_1^2=\nu_{2'}^2=1$, all twisted commutators become ordinary commutators. This means we can directly apply Prop. 4.1 of \cite{DD16} to each component of the Dirac operator (with the associated twist), but instead taking the map $(\place)^\circ$ of \cite{DD16} to mean $\nu_\ell J\pi(\place)J^{-1}\nu_\ell$ for a given $\ell\in\{1,2'\}$ (which takes the algebra to its commutant thanks to the twisted zeroth-order condition \eqn{nul0C}). 

Thus, focussing on the second component, we have that $D_{2'}$ satisfies the twisted first-order condition if and only if it decomposes into
\[ D_{2'}=D_{2',0}+D_{2',1} \]
for $D_{2',0}\in(\nu_{2'}J\pi(\Atoy)J^{-1}\nu_{2'})'$ and $D_{2',1}\in\pi(\Atoy)'$. However, we know the shape of $D_{2'}$, and so we know that no nonzero part of it lies within $\pi(\Atoy)'$, and so we must have $D_{2',1}=0$. This means we must have \begin{equation}\label{eq:gotcha}\comm{D_{2'}}{\nu_{2'}J\pi(a)J^{-1}\nu_{2'}}=0\end{equation} for all $a\in\Atoy$.

Now, by the definition of a spectral triple with multitwisted real structure we should have $\nu_{2'}JD_{2'}=\varepsilon'D_{2'}J\nu_{2'}$. However, we know that $JD_{2'}=D_{2'}J$, which implies that $D_{2'}=\pm\nu_{2'}D_{2'}\nu_{2'}$. Substituting this into \eqn{gotcha} gives
\[ \nu_{2'}\comm{D_{2'}}{J\pi(a)J^{-1}}\nu_{2'}=0, \]
but this is only true when $a$ lies within the symmetry-broken subalgebra $\CC^3$.
\end{proof}
Indeed, the above argument carries over in exactly the same manner for the full Pati-Salam case, making the appropriate replacements, \ie, replacing $D_{2'}$ by the block off-diagonal part of $\DSM$, $\Atoy$ by $\ALR$ (the symmetry-broken subalgebra of $\ALR$ being $\ASM$), and the other data by their higher-dimensional equivalents.

While this doesn't in principle rule out the 2-twisted case (since the decomposition of the Dirac operator is not unique), it does eliminate the most promising candidate. For the 3-twisted case with the decomposition we described before, Prop. \ref{env:2twistfail} carries over with only minor modifications.

\begin{prop}\label{env:3twistfail}
The requirement that
$(\Atoy,\CC^8,D_1+D_2+D_3,J,\{\nu_1, \nu_2,\nu_3\})$,
where $D_1=\left(\begin{smallmatrix}S&0\\0&\bar{S}\end{smallmatrix}\right)$, $D_2=\left(\begin{smallmatrix}0&T^*\\0&0\end{smallmatrix}\right)$ and $D_3=\left(\begin{smallmatrix}0&0\\T&0\end{smallmatrix}\right)$, with all other data defined as above and twists satisfying $\nu_1^2=\nu_2^2=\nu_3^2=1$,
be a spectral triple with multitwisted real structure breaks the algebra $\Atoy$ to $\CC^3$.
\end{prop} 
\begin{proof}
The preliminaries carry over exactly as in Prop. \ref{env:2twistfail}. Now, we focus on $D_2$ and $D_3$, beginning with $D_2$. By \cite[Prop. 4.1]{DD16}, we have that $D_{2}$ satisfies the twisted first-order condition if and only if it decomposes into 
\[ D_2=D_{2,0}+D_{2,1} \]
for $D_{2,0}\in(\nu_{2}J\pi(\Atoy)J^{-1}\nu_{2})'$ and $D_{2,1}\in\pi(\Atoy)'$. However, as before we know that  $D_{2,1}=0$. This means we must have \begin{equation}\label{eq:gotcha2}\comm{D_{2}}{\nu_{2}J\pi(a)J^{-1}\nu_{2}}=0\end{equation} for all $a\in\Atoy$.

Now, by the definition of a multitwisted real structure, we should have $\nu_{2}JD_{2}=\varepsilon'D_{2}J\nu_{2}$. However, we know that $JD_{2}J^{-1}=D_3$, which implies that $D_{2}=\pm\nu_{2}D_3\nu_{2}$. Substituting this into \eqn{gotcha2} gives
\begin{equation}\label{eq:3breaker1} \nu_{2}\comm{D_3}{J\pi(a)J^{-1}}\nu_{2}=0. \end{equation}
Going through the same procedure for $D_3$ yields 
\begin{equation}\label{eq:3breaker2} \nu_{3}\comm{D_2}{J\pi(a)J^{-1}}\nu_{3}=0, \end{equation}
and the pair of equations \eqn{3breaker1} and \eqn{3breaker2} can only be satisfied when $a$ lies within the symmetry-broken subalgebra $\CC^3$ as before.
\end{proof}

One further point which is worth remarking upon is that the Dirac operator for this toy example is not simpler than the Standard Model/Pati-Salam case only due to the lower dimensionality. With respect to a given choice of basis, one has $\DSM=\begin{pmatrix}S_{\mathrm{SM}}&T_{\mathrm{SM}}^*\\T_{\mathrm{SM}}&\bar{S}_{\mathrm{SM}}\end{pmatrix}$ where
\[ S_{\mathrm{SM}}=\begin{pmatrix}0&0&k_\nu^*&0\\0&0&0&k_e^*\\k_\nu&0&0&0\\0&k_e&0&0\end{pmatrix}\oplus
\,\bigoplus_{i=1}^3\begin{pmatrix}0&0&k_u^*&0\\0&0&0&k_d^*\\k_u&0&0&0\\0&k_d&0&0\end{pmatrix}\text{, }T_{\mathrm{SM}}=\begin{pmatrix}k_{\nu_R}&0_{1\times15}\\0_{15\times1}&0_{15\times15}\end{pmatrix}, \]
and all entries are in $M_3(\CC)$.\footnote{Note that the matrices $k_i$ ($i=\nu,e,u,d,\nu_R$) are not arbitrary elements of $M_3(\CC)$, but are subject to further constraints which we will not go into here. See \cite{CMar08} for more.} Taking $k_\nu=k_u$ and $k_e=k_d$ is called \emph{quark-lepton coupling unification}, and this simplifies the mathematics significantly. For example, if one takes $T_{\mathrm{SM}}=0$ and assumes quark-lepton coupling unification, then it is not particularly difficult to find twists (not dissimilar to the toy model case). However, even with $T_{\mathrm{SM}}=0$, without quark-leptop coupling unification this task becomes much more difficult. This is unfortunate because the model is defined at the gauge coupling unification scale, and so making such simplifying assumptions is likely to impose strong constraints on the physics up to that scale, and so ought to be avoided unless absolutely necessary.

\subsection{Other issues}

It is worth mentioning here that even apart from the above discussion, there are other issues worth mentioning in this context. Even if we had found twist(s) which recovered some twisted first-order condition, such twists would likely not be of much physical interest. The reason comes from \eqn{SFtwist} and the fact that $D_{\mathrm{SM}}$ is the fermionic mass matrix. Ordinarily, the Dirac/Majorana mass terms in the action come from the fermionic spectral action of the unfluctuated Dirac operator $$\mip{J_F\psi}{D_{\mathrm{SM}}\psi},$$ but now in the multitwisted case it seems that this should be replaced by $$\sum_\ell\mip{J_F\nu_\ell\psi}{{D}_{\ell}\psi},$$
where here $\sum_\ell D_\ell=\DSM$ specifically.

If we would like to maintain the physical relevance of the model, it would likely be necessary to instead use some 
$$J'_\ell\coloneqq\nu_\ell J_F\text{ or }D'_\ell\coloneqq\alpha_1\nu_\ell D_\ell$$ 
instead of $J_F$ or $D_\ell$ respectively (where $\nu_\ell=\alpha_1\nu_\ell^*$ for all $\ell$).\footnote{Note that Prop. \ref{env:ruinerprop} does not hold in the multitwisted formalism, but it is worth keeping track of signs nevertheless.} These both have problems though.

The choice of $J'_\ell$ seems initially preferable to $D'_\ell$, as $\nu_\ell J'_\ell D_\ell=D_\ell J'_\ell\nu_\ell$ automatically gives $J_F\DSM=\DSM J_F$ provided that $\nu_\ell=\nu_\ell^{-1}$ (for all $\ell$). However, this neatness is telling, and indeed one finds that using $J'_\ell$ reduces \eqn{nul1C} to \eqn{1C} with $J_F$ in this case, and we already know that \eqn{1C} does not hold for $\ALR$. 

If we use $D'_\ell$ instead, we run into the different issue that we already require (by definition) that $\sum_\ell D_\ell=\DSM$, but in order to have the correct action we would also need $\sum_\ell D'_\ell=\alpha_1\sum_\ell\nu_\ell D_\ell=\DSM$, which needless to say also makes it difficult to have nontrivial twists.

\section{Future directions}

The above investigation seems to leave little space for the application of twisted real structures to the spectral formulation of the left-right symmetric extension to the Standard Model, as they simply result in a reduction to the Standard Model. Is there any more that could be done? We suggest three possible avenues.

One approach could be to try to marry the (multi)twisted real structure formalism to the twisted spectral triple approach to the Standard Model and its extensions, which is an active area of research (see \eg, \cite{M15,DM17} and subsequent papers). While the cited papers focus on twisting the (doubled) commutative part of the spectral triple, it might be worthwhile in this possible `hybrid twisted' setting to investigate twists on the finite part, or even on both. We leave this long-term endeavour for future investigation. 

A second idea is to shift perspective away from the finite part of the spectral triple and towards the commutative part instead. For example, in \cite{DDM21} it was shown that the Hodge-de Rham spectral triple for a Riemannian manifold can be made into a real spectral triple if equipped with a twisted real structure. This twisted real structure may allow for nontrivial inner fluctuations of the Dirac operator $-i(d-d^*)$, which is not the case when compared to the usual spin manifold formulation with Dirac operator $\slashed{\partial}$, and that could in turn lead to some interesting geometrical features (in isolation or when coupled with a finite spectral triple).

Another approach which keeps the focus on the commutative part of the spectral triple is to attempt to follow more closely in the direction of \cite{DFLM18} and investigate if twisted real structures could have any applications to Lorentzian spectral triples, for example, by using the `untwisting' procedure described in \cite{BDS19}.  

Unrelated to these gauge theory-inspired directions, it may also prove fruitful to attempt to find twisted real structures for quantum groups which do not satisfy the first-order condition in the strict sense, for example, $\SU_q(2)$ \cite{DLSvSV05}. Though such spectral triples do not satisfy the first-order condition for a given `real structure' (only satisfying a \emph{modified} first-order condition, if any) it seems natural, and very much in the spirit of this paper, to ask if a (multi)twisted real structure could be found such that the (multi)twisted first order condition was satisfied exactly.

\section*{Acknowledgments}

We would like to thank the anonymous referee and Andrzej Sitarz for their helpful remarks and Piotr M. Hajac for his detailed comments to improve the manuscript.

\bibliography{bibliography-210608}{}
\bibliographystyle{siam}

\end{document}